\documentclass[submission,copyright,creativecommons]{eptcs}
\providecommand{\event}{QPL 2022}

\usepackage[utf8]{inputenc}
\usepackage[english]{babel}
\usepackage[T1]{fontenc}

\usepackage{hyperref}

\usepackage{mathtools}
\usepackage{amsthm}
\usepackage{amsmath}
\usepackage{amssymb}
\usepackage{mathrsfs}
\usepackage{newtxtext, newtxmath}
\usepackage{float}

\usepackage{stmaryrd}
\usepackage{braket}
\usepackage{physics}
\usepackage{tikz}
\usetikzlibrary{matrix,arrows,decorations.pathmorphing}
\usetikzlibrary{decorations.markings}
\usetikzlibrary{shapes.geometric}
\usepackage{tikz-cd}
\usepackage{tikzit}
\tikzstyle{black node}=[fill=black, draw=black, shape=circle, scale=0.3]
\tikzstyle{white node}=[fill=white, draw=black, shape=circle, scale=0.6]
\tikzstyle{triangle}=[fill=white, draw=black, regular polygon, regular polygon sides=3, scale=0.5]
\tikzstyle{big white node}=[fill=white, draw=black, shape=circle, scale=0.6]
\tikzstyle{red-edge}=[color=red]
\tikzstyle{red node}=[fill=red!60, draw=black, shape=circle, scale=0.7]
\tikzstyle{new style 0}=[fill=black, draw=black, shape=circle, scale=0.3]
\tikzstyle{green node}=[fill=zx_green, draw=black, shape=circle, scale=0.7]
\tikzstyle{hadamard}=[fill=yellow, draw=black, shape=rectangle, scale=0.7]
\tikzstyle{white-box}=[fill=white, draw=black, shape=rectangle, scale=0.5]
\tikzstyle{blue-box}=[fill=blue, draw=black, shape=rectangle, scale=0.5]
\tikzstyle{red-box}=[fill=red, draw=black, shape=rectangle, scale=0.5]

\tikzstyle{mid-arrow}=[-]
\tikzstyle{dashs}=[-, dashed, line width=0.15mm]
\tikzstyle{thick}=[-, line width=0.5mm]
\tikzstyle{arrow}=[->]
\tikzstyle{invisible}=[-, draw=none]
\tikzstyle{functor}=[-, fill={rgb,255: red,240; green,240; blue,240}]
\tikzstyle{boxedge}=[-, fill=white]
\tikzstyle{red-edge}=[-, color={blue!60}, line width=0.4mm]
\tikzstyle{arrow}=[->]
\tikzstyle{blue-edge}=[-, color={red!60}, line width=0.4mm]

\usepackage{array}

\definecolor{zxgreen}{RGB}{230,254,230}
\definecolor{zxred}{RGB}{255,135,136}
\definecolor{zxblue}{RGB}{116,116,235}
\definecolor{zxdgreen}{RGB}{91,107,91}
\definecolor{zxdred}{RGB}{142,94,94}
\definecolor{zxdblue}{RGB}{61,61,77}

\renewcommand{\tt}[1]{\mathtt{#1}}
\renewcommand{\cal}[1]{\mathcal{#1}}
\renewcommand{\phi}{\varphi}

\newcommand{\sub}{\subseteq}
\newcommand{\size}[1]{\left\vert{#1}\right\vert}
\newcommand{\bb}[1]{\mathbb{#1}}
\newcommand{\xto}[1]{\xrightarrow{#1}}

\newcommand{\B}{\cal{B}}
\newcommand{\botimes}{\hat{\otimes}}
\newcommand{\F}{\cal{F}}

\newcommand{\U}{\cal{U}}

\newtheorem{definition}{Definition}[section]

\newtheorem{example}{Example}[section]
\newtheorem{proposition}{Proposition}[section]
\newtheorem{theorem}{Theorem}[section]

\definecolor{zx_red}{RGB}{232, 165, 165}
\definecolor{zx_green}{RGB}{216, 248, 216}

\usepackage{color}
\def\bR{\begin{color}{red}}
\def\bB{\begin{color}{blue}}
\def\bM{\begin{color}{magenta}}
\def\bC{\begin{color}{cyan}}
\def\bW{\begin{color}{white}}
\def\bBl{\begin{color}{black}}
\def\bG{\begin{color}{green}}
\def\bY{\begin{color}{yellow}}
\def\e{\end{color}}

\title{Quantum Linear Optics via String Diagrams}
\author{Giovanni de Felice and Bob Coecke
\institute{Quantinuum -- Quantum Compositional Intelligence\\17 Beaumont street, OX1 2NA Oxford, UK}}

\def\titlerunning{Quantum Linear Optics via String Diagrams}
\def\authorrunning{G. de Felice \& B. Coecke}

\begin{document}
\maketitle

\begin{abstract}
     We establish a formal bridge between qubit-based
     and photonic quantum computing. We do this by defining a functor from
     the ZX calculus to linear optical circuits. In the process we provide a
     compositional theory of quantum linear optics which allows to reason about
     events involving multiple photons such as those required to perform
     linear-optical and fusion-based quantum computing.
\end{abstract}

\section{Introduction}
%!TEX root = ms.tex

Quantum optics has pioneered experimental tests of entanglement \cite{wu1950},
nonlocality \cite{aspect1982}, teleportation \cite{boschi1998},
quantum-key distribution \cite{dixon2008}, and quantum advantage \cite{zhong2020}.
These experiments ultimately rely on the ability to process coherent states of photons
in \emph{linear optical} devices, an intractable task for classical computers \cite{aaronson2010}.
Recently, the potential of using linear optics for quantum computing has
encouraged the development of both hardware \cite{corrielli2021} and software
\cite{killoran2019, perceval2022} for photonic technologies.
The first proposal was formulated by Knill, Laflamme and Millburn in 2001 \cite{knill2001a}.
Qubits are encoded in pairs of optical modes and quantum computing may be performed
using only linear optical elements and photon detectors.
Several improvements to the original scheme have been proposed in the literature
\cite{nielsen2004, kok2007, zhang2008}.
Fusion measurements were introduced by Browne and Rudolf \cite{browne2005}.
They form the basic ingredient of a recent proposal to achieve fault-tolerant
quantum computation with photonic qubits \cite{bartolucci2021}.

String diagrams provide an intuitive language for quantum processes \cite{abramsky2004, SelingerCPM, CPaqPav, CDKZ, CKbook} and are implicitly employed in quantum software packages such as tket \cite{sivarajah2021}, PyZX \cite{kissinger2019}, lambeq \cite{kartsaklis2021},
DisCoPy \cite{defelice2020b}, Quanhoven \cite{miranda2021quantum}.
On the one hand, Coecke and Duncan \cite{coecke2008} introduced the ZX calculus, a graphical
language for reasoning about qubit quantum computing, with applications
in circuit-based \cite{duncan2019}, measurement-based \cite{backens2021}, and
fault tolerant \cite{debeaudrap2020a} quantum computing.
The axioms of this calculus feature a bialgebra structure governing the $Z$ and $X$ qubit bases.
On the other hand, Vicary and Fiore used the symmetric (or bosonic) Fock space to study the
quantum harmonic oscillator, and discovered a different bialgebra structure on
this infinite dimensional Hilbert space \cite{vicary2008, fiore2015axiomatics}.
These two foundational works are hardly ever related in the literature,
possibly because of the difference in state space cardinality.
However, it is well-known that photons in linear optics behave as quantum
harmonic oscillators. Given the developments in linear-optical quantum computing,
a formal bridge should be established between qubit-based and photonic QC.
This would allow the construction of reliable software for compiling quantum
computations into photonic circuits.

In this paper, we provide such a bridge by defining a functor
from the ZX calculus to linear optics. In the process, we unify several results
on the structure and combinatorics of quantum optical experiments.
We start by studying the category of linear optical circuits, with their
classical interpretation in terms of matrices or weighted paths (Sections \ref{sec-classical}).
We then use the work of Vicary \cite{vicary2008} to
derive a functorial model for bosonic linear optics.
Our first contribution is an explicit proof that this model is equivalent to the model
based on matrix permanents of Aaronson and Arkhipov \cite{aaronson2010}
(Section \ref{sec-boson}).
Second, we introduce a graphical calculus QPath which allows to
compute the amplitudes of linear optical events involving multiple photons,
by rewriting diagrams to normal form (Section \ref{sec-qpath}).
Finally, we construct a functor from the ZX calculus to QPath and use it to describe
the basic protocols used in linear-optical and fusion-based quantum computing
(Section \ref{sec-computing}).
% \bR The theory comes with an implementation in DisCoPy \cite{defelice2020b}.
% In the appendix, we attach a Jupyter notebook that explains its features.\e

\paragraph{Related work}

Graphical approaches of linear optics are widespread in the literature. Notable
examples are the matchgates introduced by Valiant \cite{valiant2001a}, corresponding
to fermionic linear optics \cite{terhal2002b}, whose amplitudes are computed by
finding the perfect matchings of a graph. Graph-theoretic methods are also widely
used in bosonic linear optics \cite{krenn2017, ataman2018}.
There are strong links between linear optics and categorical logic.
Blute et al. \cite{blute1994} studied Fock space as exponential modality for linear logic.
The fermionic version of the Fock space has been studied in \cite{defelice2019},
it forms the W core of the ZW calculus introduced by
Coecke, Kissinger and Hadzihasanovic \cite{CK, hadzihasanovic2015diagrammatic, hadzihasanovic2017algebra}.
More recently, there has been work on a diagrammatic calculus for reasoning
about polarising beam splitters for quantum control \cite{clement2020},
an informal essay describing bosonic linear optics with category theory \cite{mccloud2022},
and a complete rewriting system for the single photon semantics of linear optical
circuits \cite{clement2022}. The ZX calculus has also been used to describe
the fault-tolerant aspects of fusion-based quantum computing \cite{bombin2021}.

% \section{Classical linear optics}\label{sec-classical}
%!TEX root = ms.tex

\section{Linear optical circuits}\label{sec-classical}

Linear optical circuits are generated by two basic physical gates.
The \emph{beam splitter} $\tt{BS}: a \otimes a \to a \otimes a$
acts on a pair of optical modes, and may be implemented using prisms or
half-silvered mirrors.
The \emph{phase shift} $\tt{S}(\alpha): a \to a$ acts on a single mode and
has a single parameter $\alpha \in [0, 2\pi]$. We depict them:
$$\scalebox{0.8}
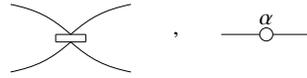
{\begin{tikzpicture}
	\begin{pgfonlayer}{nodelayer}
		\node [style=none] (12) at (-4, 0) {};
		\node [style=none] (13) at (-4, -0.25) {};
		\node [style=none] (14) at (-3, -0.25) {};
		\node [style=none] (15) at (-3, 0) {};
		\node [style=none] (16) at (-3.5, 0) {};
		\node [style=none] (17) at (-3.5, -0.25) {};
		\node [style=none] (18) at (-1.5, 1) {};
		\node [style=none] (19) at (-1.5, -1.25) {};
		\node [style=none] (20) at (-5.5, 1) {};
		\node [style=none] (21) at (-5.5, -1.25) {};
		\node [style=white node] (22) at (3, 0) {};
		\node [style=none] (23) at (1.5, 0) {};
		\node [style=none] (24) at (4.5, 0) {};
		\node [style=none] (25) at (0, 0) {$,$};
		\node [style=none] (63) at (3, 0.5) {$\alpha$};
	\end{pgfonlayer}
	\begin{pgfonlayer}{edgelayer}
		\draw (12.center) to (15.center);
		\draw (15.center) to (14.center);
		\draw (14.center) to (13.center);
		\draw (13.center) to (12.center);
		\draw [bend left=15] (20.center) to (16.center);
		\draw [bend left=15] (16.center) to (18.center);
		\draw [bend right=15] (17.center) to (19.center);
		\draw [bend left=15] (17.center) to (21.center);
		\draw (23.center) to (22);
		\draw (22) to (24.center);
	\end{pgfonlayer}
\end{tikzpicture}}$$
Linear optical circuits are obtained from these gates by composing
them vertically and horizontally. They form a set $\bf{LO}$, which has the
structure of a free monoidal category, i.e.~circuits can be composed in sequence
or in parallel.

\begin{definition}
    The \emph{classical interpretation} of $\bf{LO}$ is given by a monoidal functor
    $ \U : \bf{LO} \to \bf{Mat}_\oplus$
    into the category of matrices over the complex numbers, where $\oplus$ is
    the direct sum of vector spaces. On objects $\U$ is defined by $\U(a) = \bb{C}$.
    On arrows we have:
    $$ \U(\tt{S}(\alpha)) = (e^{i\alpha})$$
    $$ \U(\tt{BS}) = \frac{1}{\sqrt{2}} \begin{pmatrix} i & 1
                                        \\ 1 & i \end{pmatrix}$$
    where we use one standard interpretation of the beam splitter \cite{henault2015}.
\end{definition}

The Mach-Zehnder interferometer is obtained as the following composition:
$$\scalebox{0.8}
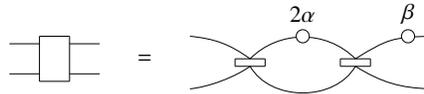
{\begin{tikzpicture}
	\begin{pgfonlayer}{nodelayer}
		\node [style=none] (0) at (-3.5, 0.75) {};
		\node [style=none] (1) at (-3.5, -0.75) {};
		\node [style=none] (2) at (-2.5, -0.75) {};
		\node [style=none] (3) at (-2.5, 0.75) {};
		\node [style=none] (4) at (-3.5, 0.5) {};
		\node [style=none] (5) at (-3.5, -0.5) {};
		\node [style=none] (6) at (3, 0) {};
		\node [style=none] (7) at (3, -0.25) {};
		\node [style=none] (8) at (4, -0.25) {};
		\node [style=none] (9) at (4, 0) {};
		\node [style=none] (10) at (3.5, 0) {};
		\node [style=none] (11) at (3.5, -0.25) {};
		\node [style=none] (12) at (7, 0) {};
		\node [style=none] (13) at (1.5, 0.75) {};
		\node [style=none] (14) at (1.5, -1) {};
		\node [style=none] (15) at (-2.5, 0.5) {};
		\node [style=none] (16) at (-2.5, -0.5) {};
		\node [style=none] (17) at (-4.5, 0.5) {};
		\node [style=none] (18) at (-4.5, -0.5) {};
		\node [style=none] (19) at (-1.5, 0.5) {};
		\node [style=none] (20) at (-1.5, -0.5) {};
		\node [style=white node] (21) at (5.25, 0.75) {};
		\node [style=none] (22) at (3.5, -0.25) {};
		\node [style=none] (23) at (7, -0.25) {};
		\node [style=none] (24) at (6.5, 0) {};
		\node [style=none] (25) at (6.5, -0.25) {};
		\node [style=none] (26) at (7.5, -0.25) {};
		\node [style=none] (27) at (7.5, 0) {};
		\node [style=none] (28) at (7, 0) {};
		\node [style=none] (29) at (7, -0.25) {};
		\node [style=none] (30) at (7, -0.25) {};
		\node [style=none] (31) at (9.5, 0.75) {};
		\node [style=none] (32) at (8.75, 0.75) {};
		\node [style=white node] (33) at (8.75, 0.75) {};
		\node [style=none] (34) at (9.5, -1) {};
		\node [style=none] (35) at (0, 0) {$=$};
		\node [style=none] (36) at (5.25, 1.5) {$2\alpha$};
		\node [style=none] (37) at (8.75, 1.5) {$\beta$};
	\end{pgfonlayer}
	\begin{pgfonlayer}{edgelayer}
		\draw (0.center) to (3.center);
		\draw (3.center) to (2.center);
		\draw (2.center) to (1.center);
		\draw (1.center) to (0.center);
		\draw (6.center) to (9.center);
		\draw (9.center) to (8.center);
		\draw (8.center) to (7.center);
		\draw (7.center) to (6.center);
		\draw [bend left=15] (13.center) to (10.center);
		\draw [bend left=15, looseness=0.75] (11.center) to (14.center);
		\draw (17.center) to (4.center);
		\draw (15.center) to (19.center);
		\draw (18.center) to (5.center);
		\draw (16.center) to (20.center);
		\draw (24.center) to (27.center);
		\draw (27.center) to (26.center);
		\draw (26.center) to (25.center);
		\draw (25.center) to (24.center);
		\draw [bend left=45] (10.center) to (28.center);
		\draw [bend right=60] (22.center) to (30.center);
		\draw [bend left=15, looseness=1.25] (28.center) to (31.center);
		\draw [bend right=15] (30.center) to (34.center);
	\end{pgfonlayer}
\end{tikzpicture}}$$
The classical interpretation of this diagram is then given by:
$$ \tt{MZI}(\alpha, \beta) = ie^{i \alpha}
\begin{pmatrix} - e^{i\beta} \tt{sin}(\alpha) & \tt{cos}(\alpha)
        \\ e^{i\beta} \tt{cos}(\alpha) & \tt{sin}(\alpha) \end{pmatrix}$$

MZIs may be used to parametrize any unitary map on $m$ modes. They are the basic
building blocks of integrated nanophotonic circuits currently being produced \cite{corrielli2021}.
The first architecture for a universal multiport interferometer was proposed
by Reck et al. \cite{reck1994} and consists of a mesh of MZIs. It was later simplified
by Clements et al. into a grid-like architecture, reducing the depth from
$d = 2m - 3$ to $d = m$ and thus the probability of photon loss \cite{clements2016}.

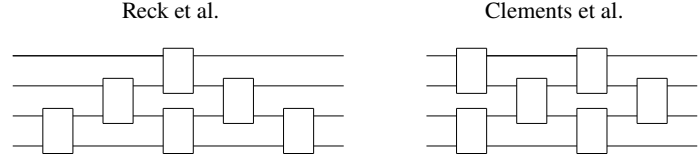
\begin{center}
    $$\scalebox{0.8}{\begin{tikzpicture}
	\begin{pgfonlayer}{nodelayer}
		\node [style=none] (12) at (1, -0.25) {};
		\node [style=none] (13) at (1, -1.75) {};
		\node [style=none] (14) at (2, -1.75) {};
		\node [style=none] (15) at (2, -0.25) {};
		\node [style=none] (16) at (1, -0.5) {};
		\node [style=none] (17) at (1, -1.5) {};
		\node [style=none] (18) at (2, -0.5) {};
		\node [style=none] (19) at (2, -1.5) {};
		\node [style=none] (20) at (0, -0.5) {};
		\node [style=none] (21) at (0, -1.5) {};
		\node [style=none] (23) at (3, -1.5) {};
		\node [style=none] (24) at (1, -2.25) {};
		\node [style=none] (25) at (1, -3.75) {};
		\node [style=none] (26) at (2, -3.75) {};
		\node [style=none] (27) at (2, -2.25) {};
		\node [style=none] (28) at (1, -2.5) {};
		\node [style=none] (29) at (1, -3.5) {};
		\node [style=none] (30) at (2, -2.5) {};
		\node [style=none] (31) at (2, -3.5) {};
		\node [style=none] (32) at (0, -2.5) {};
		\node [style=none] (33) at (0, -3.5) {};
		\node [style=none] (34) at (3, -2.5) {};
		\node [style=none] (35) at (5, -3.5) {};
		\node [style=none] (43) at (2, -0.5) {};
		\node [style=none] (47) at (5, -0.5) {};
		\node [style=none] (48) at (3, -1.25) {};
		\node [style=none] (49) at (3, -2.75) {};
		\node [style=none] (50) at (4, -2.75) {};
		\node [style=none] (51) at (4, -1.25) {};
		\node [style=none] (52) at (3, -1.5) {};
		\node [style=none] (53) at (3, -2.5) {};
		\node [style=none] (54) at (4, -1.5) {};
		\node [style=none] (55) at (4, -2.5) {};
		\node [style=none] (58) at (5, -1.5) {};
		\node [style=none] (59) at (5, -2.5) {};
		\node [style=none] (72) at (5, -0.25) {};
		\node [style=none] (73) at (5, -1.75) {};
		\node [style=none] (74) at (6, -1.75) {};
		\node [style=none] (75) at (6, -0.25) {};
		\node [style=none] (76) at (5, -0.5) {};
		\node [style=none] (77) at (5, -1.5) {};
		\node [style=none] (78) at (6, -0.5) {};
		\node [style=none] (79) at (6, -1.5) {};
		\node [style=none] (82) at (9, -0.5) {};
		\node [style=none] (83) at (7, -1.5) {};
		\node [style=none] (84) at (5, -2.25) {};
		\node [style=none] (85) at (5, -3.75) {};
		\node [style=none] (86) at (6, -3.75) {};
		\node [style=none] (87) at (6, -2.25) {};
		\node [style=none] (88) at (5, -2.5) {};
		\node [style=none] (89) at (5, -3.5) {};
		\node [style=none] (90) at (6, -2.5) {};
		\node [style=none] (91) at (6, -3.5) {};
		\node [style=none] (94) at (7, -2.5) {};
		\node [style=none] (106) at (7, -1.25) {};
		\node [style=none] (107) at (7, -2.75) {};
		\node [style=none] (108) at (8, -2.75) {};
		\node [style=none] (109) at (8, -1.25) {};
		\node [style=none] (110) at (7, -1.5) {};
		\node [style=none] (111) at (7, -2.5) {};
		\node [style=none] (112) at (8, -1.5) {};
		\node [style=none] (113) at (8, -2.5) {};
		\node [style=none] (114) at (9, -1.5) {};
		\node [style=none] (115) at (9, -2.5) {};
		\node [style=none] (116) at (9, -3.5) {};
		\node [style=none] (123) at (-13.75, -0.5) {};
		\node [style=none] (124) at (-13.75, -1.5) {};
		\node [style=none] (127) at (-10.75, -1.5) {};
		\node [style=none] (128) at (-12.75, -2.25) {};
		\node [style=none] (129) at (-12.75, -3.75) {};
		\node [style=none] (130) at (-11.75, -3.75) {};
		\node [style=none] (131) at (-11.75, -2.25) {};
		\node [style=none] (132) at (-12.75, -2.5) {};
		\node [style=none] (133) at (-12.75, -3.5) {};
		\node [style=none] (134) at (-11.75, -2.5) {};
		\node [style=none] (135) at (-11.75, -3.5) {};
		\node [style=none] (136) at (-13.75, -2.5) {};
		\node [style=none] (137) at (-13.75, -3.5) {};
		\node [style=none] (138) at (-10.75, -2.5) {};
		\node [style=none] (139) at (-8.75, -3.5) {};
		\node [style=none] (140) at (-13.75, -0.5) {};
		\node [style=none] (141) at (-8.75, -0.5) {};
		\node [style=none] (142) at (-10.75, -1.25) {};
		\node [style=none] (143) at (-10.75, -2.75) {};
		\node [style=none] (144) at (-9.75, -2.75) {};
		\node [style=none] (145) at (-9.75, -1.25) {};
		\node [style=none] (146) at (-10.75, -1.5) {};
		\node [style=none] (147) at (-10.75, -2.5) {};
		\node [style=none] (148) at (-9.75, -1.5) {};
		\node [style=none] (149) at (-9.75, -2.5) {};
		\node [style=none] (150) at (-8.75, -1.5) {};
		\node [style=none] (151) at (-8.75, -2.5) {};
		\node [style=none] (152) at (-8.75, -0.25) {};
		\node [style=none] (153) at (-8.75, -1.75) {};
		\node [style=none] (154) at (-7.75, -1.75) {};
		\node [style=none] (155) at (-7.75, -0.25) {};
		\node [style=none] (156) at (-8.75, -0.5) {};
		\node [style=none] (157) at (-8.75, -1.5) {};
		\node [style=none] (158) at (-7.75, -0.5) {};
		\node [style=none] (159) at (-7.75, -1.5) {};
		\node [style=none] (160) at (-2.75, -0.5) {};
		\node [style=none] (161) at (-6.75, -1.5) {};
		\node [style=none] (162) at (-8.75, -2.25) {};
		\node [style=none] (163) at (-8.75, -3.75) {};
		\node [style=none] (164) at (-7.75, -3.75) {};
		\node [style=none] (165) at (-7.75, -2.25) {};
		\node [style=none] (166) at (-8.75, -2.5) {};
		\node [style=none] (167) at (-8.75, -3.5) {};
		\node [style=none] (168) at (-7.75, -2.5) {};
		\node [style=none] (169) at (-7.75, -3.5) {};
		\node [style=none] (170) at (-6.75, -2.5) {};
		\node [style=none] (171) at (-6.75, -1.25) {};
		\node [style=none] (172) at (-6.75, -2.75) {};
		\node [style=none] (173) at (-5.75, -2.75) {};
		\node [style=none] (174) at (-5.75, -1.25) {};
		\node [style=none] (175) at (-6.75, -1.5) {};
		\node [style=none] (176) at (-6.75, -2.5) {};
		\node [style=none] (177) at (-5.75, -1.5) {};
		\node [style=none] (178) at (-5.75, -2.5) {};
		\node [style=none] (179) at (-2.75, -1.5) {};
		\node [style=none] (180) at (-4.75, -2.5) {};
		\node [style=none] (181) at (-4.75, -3.5) {};
		\node [style=none] (182) at (-4.75, -2.25) {};
		\node [style=none] (183) at (-4.75, -3.75) {};
		\node [style=none] (184) at (-3.75, -3.75) {};
		\node [style=none] (185) at (-3.75, -2.25) {};
		\node [style=none] (186) at (-4.75, -2.5) {};
		\node [style=none] (187) at (-4.75, -3.5) {};
		\node [style=none] (188) at (-3.75, -2.5) {};
		\node [style=none] (189) at (-3.75, -3.5) {};
		\node [style=none] (190) at (-2.75, -2.5) {};
		\node [style=none] (192) at (-2.75, -2.5) {};
		\node [style=none] (193) at (-2.75, -3.5) {};
		\node [style=none] (194) at (-8.5, 1) {Reck et al.};
		\node [style=none] (195) at (4.25, 1) {Clements et al.};
	\end{pgfonlayer}
	\begin{pgfonlayer}{edgelayer}
		\draw (12.center) to (15.center);
		\draw (15.center) to (14.center);
		\draw (14.center) to (13.center);
		\draw (13.center) to (12.center);
		\draw (20.center) to (16.center);
		\draw (21.center) to (17.center);
		\draw (19.center) to (23.center);
		\draw (24.center) to (27.center);
		\draw (27.center) to (26.center);
		\draw (26.center) to (25.center);
		\draw (25.center) to (24.center);
		\draw (32.center) to (28.center);
		\draw (30.center) to (34.center);
		\draw (33.center) to (29.center);
		\draw (31.center) to (35.center);
		\draw (43.center) to (47.center);
		\draw (43.center) to (47.center);
		\draw (48.center) to (51.center);
		\draw (51.center) to (50.center);
		\draw (50.center) to (49.center);
		\draw (49.center) to (48.center);
		\draw (54.center) to (58.center);
		\draw (55.center) to (59.center);
		\draw (72.center) to (75.center);
		\draw (75.center) to (74.center);
		\draw (74.center) to (73.center);
		\draw (73.center) to (72.center);
		\draw (78.center) to (82.center);
		\draw (79.center) to (83.center);
		\draw (84.center) to (87.center);
		\draw (87.center) to (86.center);
		\draw (86.center) to (85.center);
		\draw (85.center) to (84.center);
		\draw (90.center) to (94.center);
		\draw (106.center) to (109.center);
		\draw (109.center) to (108.center);
		\draw (108.center) to (107.center);
		\draw (107.center) to (106.center);
		\draw (112.center) to (114.center);
		\draw (113.center) to (115.center);
		\draw (91.center) to (116.center);
		\draw (124.center) to (127.center);
		\draw (128.center) to (131.center);
		\draw (131.center) to (130.center);
		\draw (130.center) to (129.center);
		\draw (129.center) to (128.center);
		\draw (136.center) to (132.center);
		\draw (134.center) to (138.center);
		\draw (137.center) to (133.center);
		\draw (135.center) to (139.center);
		\draw (140.center) to (141.center);
		\draw (140.center) to (141.center);
		\draw (142.center) to (145.center);
		\draw (145.center) to (144.center);
		\draw (144.center) to (143.center);
		\draw (143.center) to (142.center);
		\draw (148.center) to (150.center);
		\draw (149.center) to (151.center);
		\draw (152.center) to (155.center);
		\draw (155.center) to (154.center);
		\draw (154.center) to (153.center);
		\draw (153.center) to (152.center);
		\draw (158.center) to (160.center);
		\draw (159.center) to (161.center);
		\draw (162.center) to (165.center);
		\draw (165.center) to (164.center);
		\draw (164.center) to (163.center);
		\draw (163.center) to (162.center);
		\draw (168.center) to (170.center);
		\draw (171.center) to (174.center);
		\draw (174.center) to (173.center);
		\draw (173.center) to (172.center);
		\draw (172.center) to (171.center);
		\draw (177.center) to (179.center);
		\draw (178.center) to (180.center);
		\draw (169.center) to (181.center);
		\draw (182.center) to (185.center);
		\draw (185.center) to (184.center);
		\draw (184.center) to (183.center);
		\draw (183.center) to (182.center);
		\draw (188.center) to (190.center);
		\draw (189.center) to (193.center);
	\end{pgfonlayer}
\end{tikzpicture}}$$
\end{center}
Using one of these architectures, we have a parametrized circuit
$c(\theta): m \to m \in \bf{LO}$, where the parameters $\theta$
correspond to phases $\alpha, \beta$ of the Mach-Zehnder interferometers
making up the chip. As shown in both\cite{reck1994} and \cite{clements2016},
for any unitary $U: \bb{C}^m \to \bb{C}^m$,
there is a configuration of parameters $\theta$ such that
$\cal{U}(c(\theta)) = U$. We may restate their results in our notation.

\begin{proposition}[Universality] \cite{reck1994,clements2016}
    For any $m \times m$ unitary $U$, there is a circuit $c: m \to m$ in
    $\bf{LO}$ such that $\cal{U}(c) = U$.
\end{proposition}

In classical light experiments, we can measure the energy or \emph{intensity}
of an electromagnetic wave $E = E_0 e^{i(kx - \omega t)}$ where $k$ is the wavenumber,
$\omega$ is the angular frequency and $E_0$ is called the amplitude \cite{griffiths1962}.
The intensity is then given by the quadratic quantity $I = \frac{c}{2} \epsilon_0 E_0^2$
where $\epsilon_0$ is the permittivity of free space and $c$ is the speed of light.
The intensity is thus proportional to the Born rule $I \propto \norm{E}^2$.
Using the Born rule and the classical interpretation of $\bf{LO}$, we may compute the
output distribution of a photonic chip $c \in \bf{LO}$ with $m$ spatial modes,
when the input is a classical or \emph{incoherent} beam of light.
Suppose the input intensities of light are $I \in \bb{R}^m$. One may assume
$\sum_{i= 1}^m I_i = 1$.
Then the intensities $J$ at the output of an interferometer $c \in \bf{LO}$ are given by:
$$ J = \norm{\cal{U}(c)}^2 I$$
where juxtaposition denotes matrix multiplication and the norm squared $\norm{.}^2$ is
applied entry-wise. Note that $\norm{\cal{U}(c)}^2$ is a doubly stochastic
matrix since $\cal{U}(c)$ is unitary.

\begin{example}[Classical light]
    The intensities at the output of the beam splitter $\tt{BS}$
    on any normalised input $I$  are $J = (\frac{1}{2}, \frac{1}{2})$ since:
    $$ \norm{\tt{BS}}^2 = \begin{pmatrix} \frac{1}{2} & \frac{1}{2}
                            \\ \frac{1}{2} & \frac{1}{2} \end{pmatrix}$$
    The Mach-Zehnder interferometer yields the following stochastic matrix:
    $$ \norm{\tt{MZI}(\alpha, \beta)}^2 =
    \begin{pmatrix} \tt{sin}(\alpha)^2 & \tt{cos}(\alpha)^2
                 \\ \tt{cos}(\alpha)^2 & \tt{sin}(\alpha)^2 \end{pmatrix}$$
    The reflection and transmission coefficients for light intensities are
    given by $R = \tt{sin}(\alpha)^2$ and
    $T = \tt{sin}(\alpha)^2$ with $R + T = 1$.
    Thus, if we input a beam of incoherent light on the left leg $I = (1, 0)$,
    we will observe the distribution $J = (R, T)$ in the output.
\end{example}

% \section{Path calculus}\label{sec-path}
We have seen that linear optical circuits have a classical interpretation as
complex-valued matrices. We now give a graph-theoretic interpretation
of these circuits, using a syntactic category for counting paths.
The classical $\bf{Path}$ calculus has the following generators:
\begin{equation}\label{path-generators}
    \scalebox{0.8}{\begin{tikzpicture}
	\begin{pgfonlayer}{nodelayer}
		\node [style=none] (0) at (-11, 0) {};
		\node [style=none] (1) at (-10.5, 0.25) {};
		\node [style=none] (2) at (-10.5, -0.25) {};
		\node [style=none] (3) at (-12, 0) {};
		\node [style=none] (4) at (-9, 1) {};
		\node [style=none] (5) at (-9, -1) {};
		\node [style=none] (6) at (-1, 0) {};
		\node [style=none] (7) at (-1.5, 0.25) {};
		\node [style=none] (8) at (-1.5, -0.25) {};
		\node [style=none] (9) at (-3, 1) {};
		\node [style=none] (10) at (-3, -1) {};
		\node [style=none] (11) at (0, 0) {};
		\node [style=none] (12) at (-7.75, 0) {,};
		\node [style=none] (13) at (12.5, 0) {};
		\node [style=none] (14) at (15, 0) {};
		\node [style=none] (15) at (13.75, 0.5) {$r$};
		\node [style=none] (16) at (5, 0) {,};
		\node [style=none] (19) at (11, 0) {,};
		\node [style=none] (20) at (2.25, 0) {};
		\node [style=none] (21) at (3.75, 0) {};
		\node [style=white node] (22) at (2.25, 0) {};
		\node [style=none] (23) at (-5.25, 0) {};
		\node [style=none] (24) at (-6.75, 0) {};
		\node [style=white node] (25) at (-5.25, 0) {};
		\node [style=none] (26) at (-4, 0) {,};
		\node [style=none] (27) at (1, 0) {,};
		\node [style=none] (28) at (9.5, -1) {};
		\node [style=none] (29) at (9.5, 1) {};
		\node [style=none] (30) at (6.5, 1) {};
		\node [style=none] (31) at (6.5, -1) {};
	\end{pgfonlayer}
	\begin{pgfonlayer}{edgelayer}
		\draw (0.center) to (1.center);
		\draw (1.center) to (2.center);
		\draw (2.center) to (0.center);
		\draw (3.center) to (0.center);
		\draw (1.center) to (4.center);
		\draw (2.center) to (5.center);
		\draw (6.center) to (7.center);
		\draw (7.center) to (8.center);
		\draw (8.center) to (6.center);
		\draw (9.center) to (7.center);
		\draw (6.center) to (11.center);
		\draw (8.center) to (10.center);
		\draw (13.center) to (14.center);
		\draw (20.center) to (21.center);
		\draw (23.center) to (24.center);
		\draw (30.center) to (28.center);
		\draw (29.center) to (31.center);
	\end{pgfonlayer}
\end{tikzpicture}}
\end{equation}
denoted respectively $\delta$, $\epsilon$, $\mu$, $\eta$, $\sigma$ and $r$.
$\bf{Path}$ diagrams are obtained by composing these generators horizontally or
vertically. Two $\bf{Path}$ diagrams are equal if we can rewrite from one
to the other using the rules defined in Figure (\ref{path-classical}).

\begin{figure}[H]
    \centering
    \scalebox{0.9}{\input{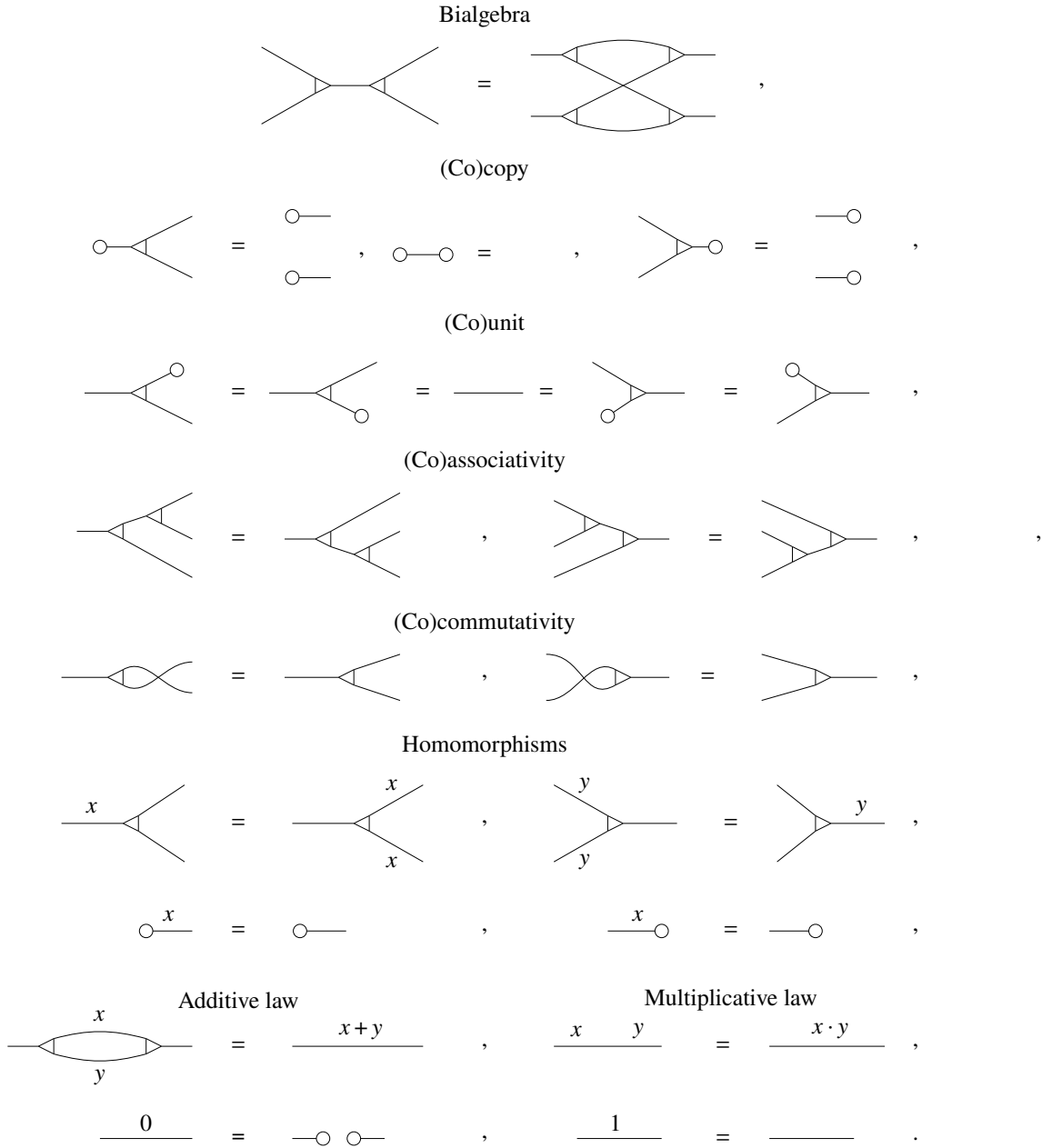}}
    \caption{Axioms of the $\bf{Path}$ calculus}
    \label{path-classical}
\end{figure}

In categorical terms, we may define $\bf{Path}$ as the PROP generated by a bialgebra
$(\delta, \epsilon, \mu, \eta)$ together with endomorphisms
$r: 1 \to 1$ with a semiring structure $r \in \bb{S}$.
Throughout this paper we fix
$\bb{S} = \bb{C}$ although our main results can be generalised to any semiring.
This calculus is folklore in category theory and was first studied by Pirashvili
\cite{pirashvili2001}. Bonchi, Sobocinski and Zanasi used it to model signal flow graphs \cite{bonchi2014}.
We can interpret $\bf{Path}$ in the monoidal category of matrices with direct sum.

\begin{proposition}
    There is a monoidal functor $\cal{C} : \bf{Path} \to \bf{Mat}_\oplus$.
\end{proposition}
\begin{proof}
    $\cal{C}$ is given on objects by $\cal{C}(a) = 1$ and on the generators
    (\ref{path-generators}) by:
    $$ \cal{C}(\delta) = \begin{pmatrix} 1 \\ 1 \end{pmatrix} \, , \quad
    \cal{C}(\epsilon) = () \, , \quad
    \cal{C}(\mu) = \begin{pmatrix} 1 & 1 \end{pmatrix} \, , \quad
    \cal{C}(\eta) = () \, , \quad
    \cal{C}(r) = \begin{pmatrix} r \end{pmatrix} \, , \quad
    \cal{C}(\sigma) = \begin{pmatrix} 0 & 1 \\ 1 & 0 \end{pmatrix} \, .$$
    where $\cal{C}(\epsilon) = () : 1 \to 0$ and $C(\eta) = () : 0 \to 1$ are the unique morphisms of that type in $\bf{Mat}_\oplus$.
    It is easy to check that all the relations in Figure \ref{path-classical}
    are satisfied by $\cal{C}$.
\end{proof}

Moreover, there is a functor turning linear optical gates into $\bf{Path}$ diagrams,
representing their underlying matrix:

$$\scalebox{0.9}
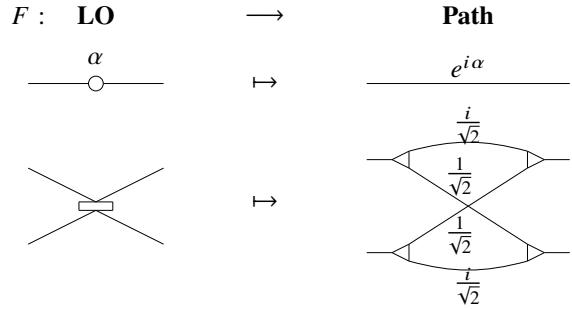
{\begin{tikzpicture}
	\begin{pgfonlayer}{nodelayer}
		\node [style=none] (77) at (3.75, 1.75) {};
		\node [style=none] (78) at (4.25, 2) {};
		\node [style=none] (79) at (4.25, 1.5) {};
		\node [style=none] (80) at (3, 1.75) {};
		\node [style=none] (81) at (8.25, 1.75) {};
		\node [style=none] (82) at (7.75, 2) {};
		\node [style=none] (83) at (7.75, 1.5) {};
		\node [style=none] (84) at (3.75, -1) {};
		\node [style=none] (85) at (4.25, -0.75) {};
		\node [style=none] (86) at (4.25, -1.25) {};
		\node [style=none] (87) at (3, -1) {};
		\node [style=none] (88) at (8.25, -1) {};
		\node [style=none] (89) at (7.75, -0.75) {};
		\node [style=none] (90) at (7.75, -1.25) {};
		\node [style=none] (91) at (9, 1.75) {};
		\node [style=none] (92) at (9, -1) {};
		\node [style=none] (93) at (-5.5, 0.5) {};
		\node [style=none] (94) at (-5.5, 0.25) {};
		\node [style=none] (95) at (-4.5, 0.25) {};
		\node [style=none] (96) at (-4.5, 0.5) {};
		\node [style=none] (97) at (-5, 0.5) {};
		\node [style=none] (98) at (-5, 0.25) {};
		\node [style=none] (99) at (-3, 1.5) {};
		\node [style=none] (100) at (-3, -0.75) {};
		\node [style=none] (101) at (-7, 1.5) {};
		\node [style=none] (102) at (-7, -0.75) {};
		\node [style=none] (108) at (0, 0.5) {$\mapsto$};
		\node [style=none] (277) at (3, 4) {};
		\node [style=none] (278) at (9, 4) {};
		\node [style=none] (279) at (6, 4.5) {$e^{i \alpha}$};
		\node [style=none] (280) at (-7, 4) {};
		\node [style=none] (281) at (-3, 4) {};
		\node [style=white node] (286) at (-5, 4) {};
		\node [style=none] (289) at (-5, 4.75) {$\alpha$};
		\node [style=none] (290) at (0, 4) {$\mapsto$};
		\node [style=none] (291) at (0, 6) {$\longrightarrow$};
		\node [style=none] (292) at (-5, 6) {$\bf{LO}$};
		\node [style=none] (293) at (6, 6) {$\bf{Path}$};
		\node [style=none] (294) at (-7, 6) {$F\, :$};
		\node [style=none] (296) at (5.75, 1.25) {$\frac{1}{\sqrt{2}}$};
		\node [style=none] (297) at (6, 2.75) {$\frac{i}{\sqrt{2}}$};
		\node [style=none] (298) at (6, -2) {$\frac{i}{\sqrt{2}}$};
		\node [style=none] (299) at (5.75, -0.5) {$\frac{1}{\sqrt{2}}$};
	\end{pgfonlayer}
	\begin{pgfonlayer}{edgelayer}
		\draw (77.center) to (78.center);
		\draw (78.center) to (79.center);
		\draw (79.center) to (77.center);
		\draw (80.center) to (77.center);
		\draw (81.center) to (82.center);
		\draw (82.center) to (83.center);
		\draw (83.center) to (81.center);
		\draw (84.center) to (85.center);
		\draw (85.center) to (86.center);
		\draw (86.center) to (84.center);
		\draw (87.center) to (84.center);
		\draw (88.center) to (89.center);
		\draw (89.center) to (90.center);
		\draw (90.center) to (88.center);
		\draw [bend left=15] (78.center) to (82.center);
		\draw (79.center) to (89.center);
		\draw (83.center) to (85.center);
		\draw [bend right=15] (86.center) to (90.center);
		\draw (88.center) to (92.center);
		\draw (81.center) to (91.center);
		\draw (93.center) to (96.center);
		\draw (96.center) to (95.center);
		\draw (95.center) to (94.center);
		\draw (94.center) to (93.center);
		\draw (101.center) to (97.center);
		\draw (97.center) to (99.center);
		\draw (98.center) to (100.center);
		\draw (98.center) to (102.center);
		\draw (277.center) to (278.center);
		\draw (280.center) to (281.center);
	\end{pgfonlayer}
\end{tikzpicture}}$$

\begin{proposition}
    The classical interpretation of linear optics factors through the Path calculus,
    i.e.~the functor $F: \bf{LO} \to \bf{Path}$ defined above satisfies
    $\U = \cal{C} \circ F$.
\end{proposition}

The rewrite rules of $\bf{Path}$ allow to reduce any diagram to a normal form,
which carries the same data as a weighted bipartite graph. This normal form can
be reached by the following (pseudo) algorithm:
\begin{enumerate}
    \item remove all possible instances of $\eta: 0 \to 1$ and $\epsilon: 1 \to 0$
    by using the (co)unit and (co)copy laws repeatedly.
    \item apply the bialgebra law, together with homomorphism and multiplicative laws, until all instances of the comonoid $\delta$ precede all instances of the monoid $\mu$,
    \item apply the additive rule to contract parallel edges.
\end{enumerate}
As an example, the following equation holds in $\bf{Path}$, the normal form
procedure going from left to right.
$$
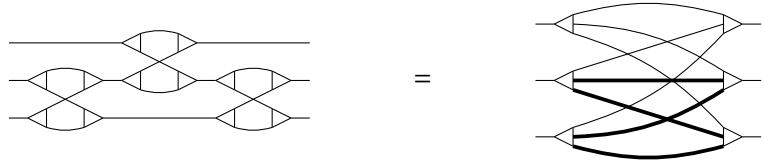
\begin{tikzpicture}
	\begin{pgfonlayer}{nodelayer}
		\node [style=none] (277) at (8.5, 1.5) {};
		\node [style=none] (278) at (8, 1.75) {};
		\node [style=none] (279) at (8, 1.25) {};
		\node [style=none] (280) at (9, 1.5) {};
		\node [style=none] (281) at (8.5, 0) {};
		\node [style=none] (282) at (8, 0.25) {};
		\node [style=none] (283) at (8, -0.25) {};
		\node [style=none] (284) at (8.5, -1.5) {};
		\node [style=none] (285) at (8, -1.25) {};
		\node [style=none] (286) at (8, -1.75) {};
		\node [style=none] (287) at (9, 0) {};
		\node [style=none] (288) at (9, -1.5) {};
		\node [style=none] (293) at (3.5, 1.5) {};
		\node [style=none] (294) at (4, 1.75) {};
		\node [style=none] (295) at (4, 1.25) {};
		\node [style=none] (296) at (3, 1.5) {};
		\node [style=none] (297) at (3.5, 0) {};
		\node [style=none] (298) at (4, 0.25) {};
		\node [style=none] (299) at (4, -0.25) {};
		\node [style=none] (300) at (3, 0) {};
		\node [style=none] (301) at (3.5, -1.5) {};
		\node [style=none] (302) at (4, -1.25) {};
		\node [style=none] (303) at (4, -1.75) {};
		\node [style=none] (304) at (3, -1.5) {};
		\node [style=none] (305) at (4, 1.5) {};
		\node [style=none] (306) at (4, 0) {};
		\node [style=none] (307) at (4, -1.5) {};
		\node [style=none] (308) at (8, -1.5) {};
		\node [style=none] (309) at (8, 0) {};
		\node [style=none] (310) at (8, 1.5) {};
		\node [style=none] (311) at (-8, 1) {};
		\node [style=none] (312) at (-7.5, 1.25) {};
		\node [style=none] (313) at (-7.5, 0.75) {};
		\node [style=none] (314) at (-6, 1) {};
		\node [style=none] (315) at (-6.5, 1.25) {};
		\node [style=none] (316) at (-6.5, 0.75) {};
		\node [style=none] (317) at (-8, 0) {};
		\node [style=none] (318) at (-7.5, 0.25) {};
		\node [style=none] (319) at (-7.5, -0.25) {};
		\node [style=none] (321) at (-6, 0) {};
		\node [style=none] (322) at (-6.5, 0.25) {};
		\node [style=none] (323) at (-6.5, -0.25) {};
		\node [style=none] (324) at (-5.5, 0) {};
		\node [style=none] (325) at (-5.5, 0) {};
		\node [style=none] (326) at (-5, 0.25) {};
		\node [style=none] (327) at (-5, -0.25) {};
		\node [style=none] (329) at (-3.5, 0) {};
		\node [style=none] (330) at (-4, 0.25) {};
		\node [style=none] (331) at (-4, -0.25) {};
		\node [style=none] (332) at (-5.5, -1) {};
		\node [style=none] (333) at (-5, -0.75) {};
		\node [style=none] (334) at (-5, -1.25) {};
		\node [style=none] (336) at (-3.5, -1) {};
		\node [style=none] (337) at (-4, -0.75) {};
		\node [style=none] (338) at (-4, -1.25) {};
		\node [style=none] (342) at (-8.5, -1) {};
		\node [style=none] (347) at (-10.5, 0) {};
		\node [style=none] (348) at (-10, 0.25) {};
		\node [style=none] (349) at (-10, -0.25) {};
		\node [style=none] (350) at (-8.5, 0) {};
		\node [style=none] (351) at (-9, 0.25) {};
		\node [style=none] (352) at (-9, -0.25) {};
		\node [style=none] (353) at (-10.5, -1) {};
		\node [style=none] (354) at (-10, -0.75) {};
		\node [style=none] (355) at (-10, -1.25) {};
		\node [style=none] (356) at (-11, -1) {};
		\node [style=none] (357) at (-8.5, -1) {};
		\node [style=none] (358) at (-9, -0.75) {};
		\node [style=none] (359) at (-9, -1.25) {};
		\node [style=none] (360) at (-8, 0) {};
		\node [style=none] (361) at (-5.5, -1) {};
		\node [style=none] (365) at (-11, -1) {};
		\node [style=none] (366) at (-3, -1) {};
		\node [style=none] (367) at (-3, 0) {};
		\node [style=none] (368) at (-11, 0) {};
		\node [style=none] (369) at (-11, 1) {};
		\node [style=none] (370) at (-3, 1) {};
		\node [style=none] (371) at (0, 0) {$=$};
	\end{pgfonlayer}
	\begin{pgfonlayer}{edgelayer}
		\draw (277.center) to (278.center);
		\draw (278.center) to (279.center);
		\draw (279.center) to (277.center);
		\draw (277.center) to (280.center);
		\draw (281.center) to (282.center);
		\draw (282.center) to (283.center);
		\draw (283.center) to (281.center);
		\draw (284.center) to (285.center);
		\draw (285.center) to (286.center);
		\draw (286.center) to (284.center);
		\draw (284.center) to (288.center);
		\draw (281.center) to (287.center);
		\draw (293.center) to (294.center);
		\draw (294.center) to (295.center);
		\draw (295.center) to (293.center);
		\draw (296.center) to (293.center);
		\draw (297.center) to (298.center);
		\draw (298.center) to (299.center);
		\draw (299.center) to (297.center);
		\draw (300.center) to (297.center);
		\draw (301.center) to (302.center);
		\draw (302.center) to (303.center);
		\draw (303.center) to (301.center);
		\draw (304.center) to (301.center);
		\draw [bend left=15] (294.center) to (278.center);
		\draw (298.center) to (310.center);
		\draw [bend right=15] (302.center) to (279.center);
		\draw [bend left=15] (305.center) to (282.center);
		\draw [bend left=15] (295.center) to (285.center);
		\draw (311.center) to (312.center);
		\draw (312.center) to (313.center);
		\draw (313.center) to (311.center);
		\draw (314.center) to (315.center);
		\draw (315.center) to (316.center);
		\draw (316.center) to (314.center);
		\draw (317.center) to (318.center);
		\draw (318.center) to (319.center);
		\draw (319.center) to (317.center);
		\draw (321.center) to (322.center);
		\draw (322.center) to (323.center);
		\draw (323.center) to (321.center);
		\draw [bend left=15] (312.center) to (315.center);
		\draw (313.center) to (322.center);
		\draw (316.center) to (318.center);
		\draw [bend right=15] (319.center) to (323.center);
		\draw (321.center) to (324.center);
		\draw (325.center) to (326.center);
		\draw (326.center) to (327.center);
		\draw (327.center) to (325.center);
		\draw (329.center) to (330.center);
		\draw (330.center) to (331.center);
		\draw (331.center) to (329.center);
		\draw (332.center) to (333.center);
		\draw (333.center) to (334.center);
		\draw (334.center) to (332.center);
		\draw (336.center) to (337.center);
		\draw (337.center) to (338.center);
		\draw (338.center) to (336.center);
		\draw [bend left=15] (326.center) to (330.center);
		\draw (327.center) to (337.center);
		\draw (331.center) to (333.center);
		\draw [bend right=15] (334.center) to (338.center);
		\draw (347.center) to (348.center);
		\draw (348.center) to (349.center);
		\draw (349.center) to (347.center);
		\draw (350.center) to (351.center);
		\draw (351.center) to (352.center);
		\draw (352.center) to (350.center);
		\draw (353.center) to (354.center);
		\draw (354.center) to (355.center);
		\draw (355.center) to (353.center);
		\draw (356.center) to (353.center);
		\draw (357.center) to (358.center);
		\draw (358.center) to (359.center);
		\draw (359.center) to (357.center);
		\draw [bend left=15] (348.center) to (351.center);
		\draw (349.center) to (358.center);
		\draw (352.center) to (354.center);
		\draw [bend right=15] (355.center) to (359.center);
		\draw (357.center) to (361.center);
		\draw (350.center) to (360.center);
		\draw (368.center) to (347.center);
		\draw (329.center) to (367.center);
		\draw (366.center) to (336.center);
		\draw (369.center) to (311.center);
		\draw (314.center) to (370.center);
		\draw [style=thick] (306.center) to (309.center);
		\draw [style=thick] (299.center) to (308.center);
		\draw [style=thick, bend left=15] (286.center) to (303.center);
		\draw [style=thick, bend right=15] (307.center) to (283.center);
	\end{pgfonlayer}
\end{tikzpicture}$$
where the thick wires carry the endomorphism $2: 1 \to 1$.
Computation of the weights on the resulting graph is equivalent to the
block-diagonal matrix multiplication defined by $\cal{C}$. This is stated
formally as the following result.

\begin{proposition}[Completeness]
    The axioms of $\bf{Path}$ are complete for $\bf{Mat}_\oplus$,
    i.e. $\cal{C}: \bf{Path} \to \bf{Mat}_\oplus$ is a monoidal equivalence.
\end{proposition}
\begin{proof}
    The normalisation procedure is described above, see also \cite[Proposition 1]{bonchi2014}.
\end{proof}

\section{Fock space and permanents}\label{sec-boson}
%!TEX root = ms.tex

Processing bosonic particles, such as photons, with linear optical devices
gives rise to statistics that are hard to simulate classically \cite{aaronson2010}.
In this section we give an interpretation of linear optical circuits, derived
from \cite{vicary2008}, in terms of free and symmetric Fock space functors
$\F, \B : \bf{Mat}_\oplus \to \bf{Vect}_\otimes$.
We show that this characterisation is equivalent to the model introduced in \cite{aaronson2010}.

Consider a box containing particles. Assume that the space of states of a single particle
is given by a Hilbert space $H$. The \emph{free} Fock space is defined as follows:
$$\F(H) = \bigoplus_{n = 0}^\infty H^{\otimes n} $$
where $\otimes$ is the usual tensor product and $\oplus$ the direct sum.
$\F(H)$ describes the state space of a given number of \emph{distinguishable}
particles indexed by $n$.
Given a basis $X$ of \emph{modes} such that $H = \bb{C}X$, we have that:
$$\F(\bb{C}X) \simeq l^2(X^\ast)$$
where $\bb{C}X$ denotes the free vector space space with basis $X$,
$X^\ast$ is the set of lists over $X$ and $l^2$ is the canonical Hilbert space construction defined in \cite{heunen2013}. Thus for $n$ particles in $m$ modes we have
$\F_n(\bb{C}^m) = (\bb{C}^m)^{\otimes n} \simeq \bb{C}([m]^n)$
the basis states $[m]^n$ are given by lists of length $n$ using $m$ distinct symbols.

\begin{proposition}
   The free Fock space can be extended to a functor $\F : \bf{Mat}_\oplus \to \bf{Vect}_\otimes$ defined on the $n$-particle sector by
   $\F_n(A) = A^{\otimes n}$ for matrices $A$.
\end{proposition}
\begin{proof}
    This follows by functoriality of tensor $\otimes$ and biproduct $\oplus$.
\end{proof}

Now suppose that the particles in the box are \emph{indistinguishable}. The state
space of the system will then be described by the symmetric or \emph{bosonic}
Fock space, defined as follows:
$$\B(H) = \bigoplus_{n = 0}^\infty H^{\botimes n} $$
where $\botimes$ is the quotient of the tensor product by the equivalence relation
$x \botimes y = y \botimes x$, which ensures that the bosons are indistinguishable.
One may show that $\B(\bb{C}X) \simeq l^2(\bb{N}^X)$, i.e. the bosonic Fock space
over a set of modes $X$ is spanned by the basis states of occupation numbers.
The $n$-particle sector of the bosonic Fock space $\B_n(H)$ is the
$n$-th component in the direct sum above.
When $H = \bb{C}^m$ has dimension $m$, we have $n$ indistinguishable particles
in $m$ possible modes.
The basis states of $\B_n(H)$ are given by:
$$ \Phi_{m, n} = \set{(s_1, \dots, s_m)\, \vert\, \sum_{i=1}^m s_i = n \, , \, s_i \in \bb{N}} \, \sub \, \bb{N}^m$$
Note that $\vert \Phi_{m, n} \vert = \binom{m + n - 1}{n}$ and
$\B_n(H) =  H^{\botimes n} = \bb{C}(\Phi_{m, n})$.
Let us compare the basis states for distinguishable and bosonic particles.
There is a family of linear maps $\alpha_H: \F(H) \to \B(H)$ defined on the
basis states of the $n$-particle sector $X \in [m]^n$ by:
$$ \alpha(X) = \sqrt{\frac{n!}{\prod_{j = 1}^m a(X)_j!}} \, \ket{a(X)}$$
where $a : [m]^n \to \Phi_{m, n}$ is defined by
$a(X)_j= \size{\set{i \vert X_i = j}}$ for $j \in [m]$. Note that the normalisation
factor is equal to the size of the pre-image $a^{-1}(a(X))$.
Let us write the map $\alpha^\dagger$ explicitly:
$$ \alpha^\dagger \ket{I} = \sqrt{\frac{N_I}{n!}}\sum_{X \in a^{-1}(I)} \ket{X}$$
where $N_I = \prod_{j=1}^m I_j!$.
We can now use $\alpha$ to define the action of $\B$ on arrows.

\begin{proposition} \cite{vicary2008}
   The bosonic Fock space can be extended to a strong monoidal functor $\B : \bf{Mat}_\oplus \to \bf{Vect}_\otimes$
   defined on arrows $A : m \to k$ by:
   $$\B_n(A) \, = \, A^{\botimes n}\, =\, \alpha A^{\otimes n} \alpha^\dagger$$
   and satisfying $\B(A \oplus B) = \B(A) \otimes \B(B)$.
\end{proposition}
\begin{proof}
    Functoriality follows from naturality of $\alpha^\dagger \alpha$ \cite[Lemma 6.6]{vicary2008}. $\B$ is moreover strong monoidal:
    $$ \B(\bb{C}X \oplus \bb{C}Y) \simeq \B(\bb{C}(X + Y))
    \simeq l^2(\bb{N}^{X + Y})
    \simeq l^2(\bb{N}^X \times \bb{N}^Y)
    \simeq \B(\bb{C}X) \otimes \B(\bb{C}Y)$$
\end{proof}

We can use the bosonic Fock space to define a functorial model for linear optics.

\begin{definition}[Functorial model]
    The functorial interpretation of linear optics is given by the composition
    $\B : \bf{LO} \xto{\U} \bf{Mat}_\oplus \xto{\B} \bf{Vect}_\otimes$.
    Given a chip $c: m \to m \in \bf{LO}$, the probability of observing output state
    $J \in \Phi_{m, n}$ on input $I \in \Phi_{m, n}$ is given by:
    $$ P_c^\B(J \vert I) = \norm{\bra{J}\B(c)\ket{I}}^2 = \norm{\bra{J} \alpha \, \U(c)^{\otimes n}\, \alpha^\dagger \ket{I}}^2$$
\end{definition}

Aaronson and Arkhipov \cite{aaronson2010} introduced a formal model for linear
optics based on matrix permanents.

\begin{definition}[Permanent model \cite{aaronson2010}]
    Given a chip $c: m \to m \in \bf{LO}$, the probability of observing output state
    $J \in \Phi_{m, n}$ on input $I \in \Phi_{m, n}$ is given by:
    $$ P_c(J \vert I) = \frac{1}{N_I N_J} \norm{\mathtt{Perm}(\U(c)_{I, J})}^2$$
    where $N_S = \prod_{j=1}^m S_j!$, $\tt{Perm}$ denotes the matrix permanent,
    and $U_{I, J}$ is the $n \times n$ matrix obtained from an $m \times m$
    matrix $U$ as follows.
    We first construct the $m \times n$ matrix $U_J$ by taking $J_j$ copies of
    the $j$th column of $U$ for each $j \leq m$. Then we construct $U_{I, J}$
    by taking $I_i$ copies of the $i$th row of $U_J$.
\end{definition}

We give an explicit proof that the models introduced above are equivalent, although
the argument can be traced back to Fock \cite{fock1932}.

\begin{theorem}
    The functorial model of linear optics is equivalent to the permanent model.
    Explicitly, for any $m \times m$ unitary $U$ and basis states $I, J \in \Phi_{m, n}$
    $$ \bra{J} \B(U) \ket{I} = \frac{\mathtt{Perm}(U_{I, J})}{\sqrt{N_I N_J}}$$
\end{theorem}
\begin{proof}
    We start by expanding the left-hand side:
    \begin{align*}
        \bra{J} \B(U) \ket{I} &= \bra{J} U^{\botimes n} \ket{I} = \bra{J} \alpha U^{\otimes n} \alpha^\dagger \ket{I}
        = (\alpha^\dagger \ket{J})^\dagger U^{\otimes n} (\alpha^\dagger \ket{I})\\
        &= \left( \sqrt{\frac{N_J}{n!}} \sum_{Y \in a^{-1}(J)} \bra{Y} \right) U^{\otimes n} \, \left(\sqrt{\frac{N_I}{n!}} \sum_{X \in a^{-1}(I)} \ket{X} \right)\\
        &= \frac{\sqrt{N_J N_I}}{n!} \size{a^{-1}(I)} \sum_{Y \in a^{-1}(J)} \bra{Y} U^{\otimes n} \ket{\hat{X}}\\
        &=  \frac{\sqrt{N_J N_I}}{n!} \frac{n!}{N_I} \frac{1}{N_J} \sum_{\sigma \in S_n} \prod_{i = 1}^n U_{\hat{X}_i, \hat{Y}_{\sigma(i)}}\\
        &=  \frac{1}{\sqrt{N_I N_J}} \mathtt{Perm}(U_{I, J})
    \end{align*}
    where $\hat{X} \in a^{-1}(I)$ and $\hat{Y} \in a^{-1}(J)$ are any chosen representatives.
    Note that this choice is irrelevant since we sum over all permutations, and
    so in particular we can set $(U_{I, J})_{i j} = U_{\hat{X}_i, \hat{Y}_j}$, yielding the last step.
\end{proof}

\begin{example}[Hong-Ou-Mandel]\label{ex-hong-ou-mandel}
    Consider the matrix of the beam splitter:
    $$ U = \frac{1}{\sqrt{2}} \begin{pmatrix} i & 1 \\ 1 & i\end{pmatrix}$$
    Suppose we input one boson in each port $I = (1, 1)$.
    There are three possible outcomes $J =\, (2, 0) , \, (1, 1) , \, (0, 2)$.
    We may determine the amplitudes of the different outcomes by computing permanents:
    $$ \mathtt{Perm}\begin{pmatrix} i & i \\ 1 & 1\end{pmatrix} = 2i
    \quad \mathtt{Perm} \begin{pmatrix} i & 1 \\ 1 & i\end{pmatrix} = 0
    \quad \mathtt{Perm}\begin{pmatrix} 1 & 1 \\ i & i\end{pmatrix} = 2i$$
    The component for outcome $(1, 1)$ is $0$. We deduce that the
    probability of observing one boson in each output port is $0$. Thus
    interference ensures that the bosons bunch together at the output of the device,
    a phenomenon known as the Hong-Ou-Mandel effect.
\end{example}

\section{Quantum paths and matchings}\label{sec-qpath}
%!TEX root = ms.tex

In the previous section we have shown that bosonic linear optics can be
formulated equivalently in terms of Fock space and permanents.
Aaronson and Arkhipov \cite{aaronson2010} used the second definition to show
that sampling from a linear
optical chip with bosonic particles is classically hard: if a classical
computer can compute an additive approximation of matrix permanents then the
polynomial hierarchy collapses.
While this computational definition is useful for proving complexity
results, we want to develop a diagrammatic syntax for programming linear
optical circuits.
We do this by developing a quantised calculus $\bf{QPath}$ which allows us to
compute the amplitudes of linear optical events involving multiple photons,
using simple rewrite rules.
In order to quantise the $\bf{Path}$ calculus, we add creation and
annihilation of particles as generators.
$$\mathbf{QPath} = \mathbf{Path} + \left\{ \scalebox{0.9}{\begin{tikzpicture}
	\begin{pgfonlayer}{nodelayer}
		\node [style=black node] (118) at (2.5, 0) {};
		\node [style=none] (119) at (1, 0) {};
		\node [style=black node] (122) at (-2.5, 0) {};
		\node [style=none] (123) at (-1, 0) {};
		\node [style=none] (124) at (0, 0) {,};
		\node [style=none] (125) at (-2.5, 0.75) {$\ket{n}$};
		\node [style=none] (126) at (2.5, 0.75) {$\ket{n}$};
	\end{pgfonlayer}
	\begin{pgfonlayer}{edgelayer}
		\draw (118) to (119.center);
		\draw (122) to (123.center);
	\end{pgfonlayer}
\end{tikzpicture}} \right\} _{n \in \bb{N}^+}$$
This yields a free monoidal category where we can represent linear optical
processes with state preparations (creation) and post-selection (annihilation).
Before developing a calculus around the $\bf{QPath}$ generators, the first thing
to note is that $\bf{QPath}$ is equivalent to $\bf{Path}$ if we interpret it
classically, i.e. functors $\bf{QPath} \to \bf{Mat}_\oplus$ are
in bijective correspondence with functors $\bf{Path} \to \bf{Mat}_\oplus$.
In fact, black and white nodes are necessarily equal in $\bf{Mat}_\oplus$,
since the unit $0$ is both a terminal and an initial object.
In order to interpret black nodes, representing modes occupied by photons,
we need to use the bosonic Fock space functor.

The quantum interpretation $\B : \bf{QPath} \to \bf{Hilb}_\otimes$ is obtained
on the Path generators (\ref{path-generators}) by composing
$\cal{C}: \bf{Path} \to \bf{Mat}_\oplus$ with the
bosonic Fock space functor $\B: \bf{Mat}_\oplus \to \bf{Hilb}_\otimes$.
The generating object $a$ of $\bf{QPath}$ is mapped to the free Hilbert space
$l^2(\bb{N})$. The comonoid $\delta : 1 \to 2$ is mapped as follows:
$$\scalebox{0.7}{\begin{tikzpicture}
	\begin{pgfonlayer}{nodelayer}
		\node [style=none] (0) at (-0.5, 0) {};
		\node [style=none] (1) at (0, 0.25) {};
		\node [style=none] (2) at (0, -0.25) {};
		\node [style=none] (3) at (2, 1.5) {};
		\node [style=none] (4) at (2, -1.5) {};
		\node [style=none] (5) at (-2.5, 0) {};
	\end{pgfonlayer}
	\begin{pgfonlayer}{edgelayer}
		\draw (0.center) to (1.center);
		\draw (1.center) to (2.center);
		\draw (2.center) to (0.center);
		\draw (1.center) to (3.center);
		\draw (2.center) to (4.center);
		\draw (5.center) to (0.center);
	\end{pgfonlayer}
\end{tikzpicture}} \qquad \mapsto \qquad \B(\delta)\,:\, \ket{n} \quad \mapsto \quad \sum_{k=0}^n \binom{n}{k}^{\frac{1}{2}} \ket{k}\ket{n - k}\, ,$$
while the monoid $\mu$ is mapped to the \emph{dagger} $\B(\mu) = \B(\delta)^\dagger$.
White nodes are mapped to $\ket{0}$, $\bra{0}$, indicating that the mode is empty.
Endomorphisms $r: 1 \to 1$ in $\bf{QPath}$ are interpreted as follows:
$$\scalebox{0.8}{\begin{tikzpicture}
	\begin{pgfonlayer}{nodelayer}
		\node [style=none] (0) at (2, 0) {};
		\node [style=none] (1) at (-2, 0) {};
		\node [style=none] (4) at (0, 0.5) {$r$};
	\end{pgfonlayer}
	\begin{pgfonlayer}{edgelayer}
		\draw (1.center) to (0.center);
	\end{pgfonlayer}
\end{tikzpicture}
} \qquad \mapsto \qquad \B(r) \,:\, \ket{n}  \quad \mapsto \quad r^n \ket{n}$$
Finally, the black nodes in $\bf{QPath}$ are mapped respectively to
$\ket{n}$ and $\bra{n}$, indicating that the mode is occupied by $n$ particles.
Hadzihasanovic \cite{hadzihasanovic2017algebra} showed directly that
$(\mu, \ket{0}, \delta, \bra{0})$ forms a bialgebra. In fact all the axioms
that hold in the classical interpretation $\cal{C}$ also hold
in the bosonic interpretation $\cal{B}$ since it is defined by functor-composition.
However, black nodes allow to express some processes which were not available
in the classical semantics, as we will see below.

The axioms of $\bf{QPath}$ include all the axioms of $\bf{Path}$, given in Figure
\ref{path-classical}.
The only additional rules we will need to reason with black nodes are the following:
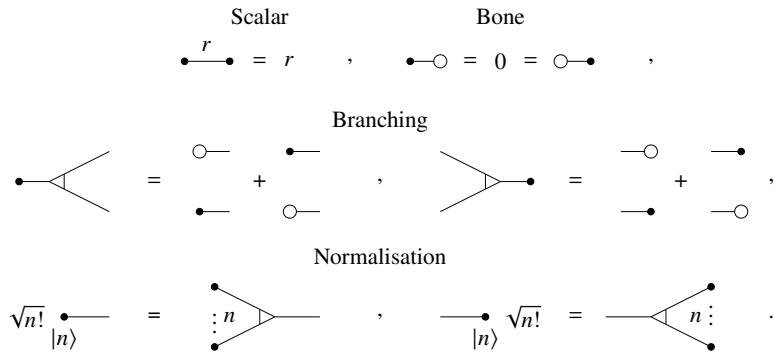
\begin{figure}[H]
    \centering
    \scalebox{0.8}{\begin{tikzpicture}
	\begin{pgfonlayer}{nodelayer}
		\node [style=none] (0) at (-11, 1) {};
		\node [style=none] (1) at (-10.5, 1.25) {};
		\node [style=none] (2) at (-10.5, 0.75) {};
		\node [style=none] (4) at (-9, 2) {};
		\node [style=none] (7) at (-9, 0) {};
		\node [style=none] (12) at (-12, 1) {};
		\node [style=black node] (14) at (-12, 1) {};
		\node [style=none] (15) at (-5, 2) {};
		\node [style=none] (16) at (-6, 2) {};
		\node [style=none] (18) at (-5, 0) {};
		\node [style=none] (19) at (-6, 0) {};
		\node [style=black node] (20) at (-6, 0) {};
		\node [style=none] (21) at (-7.5, 1) {$=$};
		\node [style=white node] (22) at (-6, 2) {};
		\node [style=none] (23) at (-2, 0) {};
		\node [style=none] (24) at (-3, 0) {};
		\node [style=none] (25) at (-2, 2) {};
		\node [style=none] (26) at (-3, 2) {};
		\node [style=black node] (27) at (-3, 2) {};
		\node [style=white node] (28) at (-3, 0) {};
		\node [style=none] (29) at (-4, 1) {$+$};
		\node [style=none] (37) at (9, 2) {};
		\node [style=none] (38) at (8, 2) {};
		\node [style=none] (39) at (9, 0) {};
		\node [style=none] (40) at (8, 0) {};
		\node [style=black node] (41) at (9, 0) {};
		\node [style=none] (42) at (6.5, 1) {$=$};
		\node [style=white node] (43) at (9, 2) {};
		\node [style=none] (44) at (12, 0) {};
		\node [style=none] (45) at (11, 0) {};
		\node [style=none] (46) at (12, 2) {};
		\node [style=none] (47) at (11, 2) {};
		\node [style=black node] (48) at (12, 2) {};
		\node [style=white node] (49) at (12, 0) {};
		\node [style=none] (50) at (10, 1) {$+$};
		\node [style=none] (52) at (2, 2) {};
		\node [style=none] (54) at (4, 1) {};
		\node [style=none] (55) at (3.5, 1.25) {};
		\node [style=none] (56) at (3.5, 0.75) {};
		\node [style=none] (57) at (2, 0) {};
		\node [style=none] (73) at (5, 1) {};
		\node [style=black node] (75) at (5, 1) {};
		\node [style=none] (76) at (0, 1) {,};
		\node [style=none] (77) at (0, 3) {Branching};
		\node [style=none] (78) at (2, 5) {};
		\node [style=none] (79) at (1, 5) {};
		\node [style=black node] (80) at (1, 5) {};
		\node [style=white node] (81) at (2, 5) {};
		\node [style=none] (82) at (3, 5) {$=$};
		\node [style=none] (83) at (4, 5) {$0$};
		\node [style=none] (84) at (-5, 5) {};
		\node [style=none] (85) at (-6.5, 5) {};
		\node [style=black node] (86) at (-6.5, 5) {};
		\node [style=black node] (87) at (-5, 5) {};
		\node [style=none] (88) at (-4, 5) {$=$};
		\node [style=none] (89) at (-3, 5) {$r$};
		\node [style=none] (90) at (-1, 5) {,};
		\node [style=none] (91) at (-5.75, 5.5) {$r$};
		\node [style=none] (92) at (-4, 6.5) {Scalar};
		\node [style=none] (94) at (7, 5) {};
		\node [style=none] (95) at (6, 5) {};
		\node [style=black node] (96) at (7, 5) {};
		\node [style=white node] (97) at (6, 5) {};
		\node [style=none] (98) at (5, 5) {$=$};
		\node [style=none] (99) at (0, -1.5) {Normalisation};
		\node [style=none] (100) at (4, 6.5) {Bone};
		\node [style=none] (101) at (0, -3.5) {,};
		\node [style=none] (102) at (9, -3.5) {};
		\node [style=none] (103) at (9.5, -3.25) {};
		\node [style=none] (104) at (9.5, -3.75) {};
		\node [style=none] (105) at (11, -2.5) {};
		\node [style=none] (106) at (11, -4.5) {};
		\node [style=black node] (107) at (11, -2.5) {};
		\node [style=black node] (108) at (11, -4.5) {};
		\node [style=none] (109) at (7.5, -3.5) {};
		\node [style=none] (110) at (-5.5, -3.5) {$\vdots$};
		\node [style=none] (111) at (-5, -3.5) {$n$};
		\node [style=black node] (112) at (3.5, -3.5) {};
		\node [style=none] (113) at (2, -3.5) {};
		\node [style=none] (114) at (6.5, -3.5) {$=$};
		\node [style=none] (115) at (-3.5, -3.5) {};
		\node [style=none] (116) at (-4, -3.25) {};
		\node [style=none] (117) at (-4, -3.75) {};
		\node [style=none] (118) at (-2, -3.5) {};
		\node [style=black node] (119) at (-5.5, -2.5) {};
		\node [style=none] (120) at (-5.5, -2.5) {};
		\node [style=none] (121) at (-5.5, -4.5) {};
		\node [style=none] (122) at (3.5, -4.25) {$\ket{n}$};
		\node [style=black node] (123) at (-5.5, -4.5) {};
		\node [style=none] (124) at (-7.5, -3.5) {=};
		\node [style=none] (125) at (11, -3.25) {$\vdots$};
		\node [style=none] (126) at (10.5, -3.5) {$n$};
		\node [style=black node] (127) at (-10.5, -3.5) {};
		\node [style=none] (128) at (-9, -3.5) {};
		\node [style=none] (129) at (-10.5, -4.25) {$\ket{n}$};
		\node [style=none] (130) at (13, -3.5) {.};
		\node [style=none] (132) at (-11.75, -3.5) {$\sqrt{n!}$};
		\node [style=none] (133) at (4.75, -3.5) {$\sqrt{n!}$};
		\node [style=none] (134) at (13, 1) {,};
		\node [style=none] (135) at (9, 5) {,};
	\end{pgfonlayer}
	\begin{pgfonlayer}{edgelayer}
		\draw (0.center) to (1.center);
		\draw (1.center) to (2.center);
		\draw (2.center) to (0.center);
		\draw (1.center) to (4.center);
		\draw (2.center) to (7.center);
		\draw (12.center) to (0.center);
		\draw (16.center) to (15.center);
		\draw (19.center) to (18.center);
		\draw (24.center) to (23.center);
		\draw (26.center) to (25.center);
		\draw (38.center) to (37.center);
		\draw (40.center) to (39.center);
		\draw (45.center) to (44.center);
		\draw (47.center) to (46.center);
		\draw (54.center) to (55.center);
		\draw (55.center) to (56.center);
		\draw (56.center) to (54.center);
		\draw (52.center) to (55.center);
		\draw (56.center) to (57.center);
		\draw (54.center) to (73.center);
		\draw (79.center) to (78.center);
		\draw (85.center) to (84.center);
		\draw (95.center) to (94.center);
		\draw (102.center) to (103.center);
		\draw (103.center) to (104.center);
		\draw (104.center) to (102.center);
		\draw (104.center) to (106.center);
		\draw (103.center) to (105.center);
		\draw (109.center) to (102.center);
		\draw (113.center) to (112);
		\draw (115.center) to (116.center);
		\draw (116.center) to (117.center);
		\draw (117.center) to (115.center);
		\draw (118.center) to (115.center);
		\draw (121.center) to (117.center);
		\draw (116.center) to (120.center);
		\draw (128.center) to (127);
	\end{pgfonlayer}
\end{tikzpicture}}
    \caption{Additional axioms for the $\bf{QPath}$ calculus}
    \label{path-quantum}
\end{figure}
It is easy to show that the axioms above are sound for the bosonic interpretation $\B$.

\begin{example}[Creation/Annihilation]
    The creation and annihilation operators on single modes have the following
    representation as $\bf{QPath}$ diagrams.
    $$\scalebox{0.8}
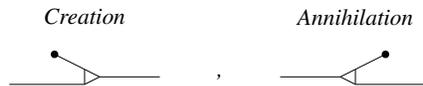
{\begin{tikzpicture}
	\begin{pgfonlayer}{nodelayer}
		\node [style=none] (0) at (-5.5, 0.75) {};
		\node [style=none] (1) at (-4, 0) {};
		\node [style=none] (2) at (-4.5, 0.25) {};
		\node [style=none] (3) at (-4.5, -0.25) {};
		\node [style=none] (4) at (-7, -0.25) {};
		\node [style=none] (5) at (-2, 0) {};
		\node [style=black node] (6) at (-5.5, 0.75) {};
		\node [style=none] (7) at (4, 0) {};
		\node [style=none] (8) at (4.5, 0.25) {};
		\node [style=none] (9) at (4.5, -0.25) {};
		\node [style=none] (10) at (5.5, 0.75) {};
		\node [style=none] (11) at (7, -0.25) {};
		\node [style=none] (12) at (2, 0) {};
		\node [style=black node] (13) at (5.5, 0.75) {};
		\node [style=none] (14) at (0, 0) {,};
		\node [style=none] (15) at (-4.5, 2) {Creation};
		\node [style=none] (16) at (4.5, 2) {Annihilation};
	\end{pgfonlayer}
	\begin{pgfonlayer}{edgelayer}
		\draw (1.center) to (2.center);
		\draw (2.center) to (3.center);
		\draw (3.center) to (1.center);
		\draw (0.center) to (2.center);
		\draw (3.center) to (4.center);
		\draw (1.center) to (5.center);
		\draw (7.center) to (8.center);
		\draw (8.center) to (9.center);
		\draw (9.center) to (7.center);
		\draw (8.center) to (10.center);
		\draw (9.center) to (11.center);
		\draw (12.center) to (7.center);
	\end{pgfonlayer}
\end{tikzpicture}}$$
    We recover the commuting relations for these operators using the branching law:
    $$\scalebox{0.8}
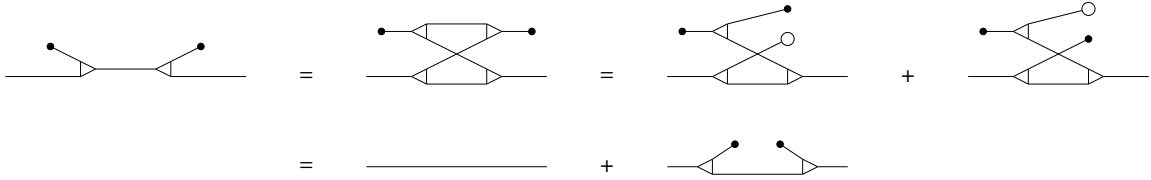
{\begin{tikzpicture}
	\begin{pgfonlayer}{nodelayer}
		\node [style=none] (0) at (-8.5, 1) {};
		\node [style=none] (1) at (-7, 0.25) {};
		\node [style=none] (2) at (-7.5, 0.5) {};
		\node [style=none] (3) at (-7.5, 0) {};
		\node [style=none] (4) at (-10, 0) {};
		\node [style=black node] (6) at (-8.5, 1) {};
		\node [style=none] (7) at (-5, 0.25) {};
		\node [style=none] (8) at (-4.5, 0.5) {};
		\node [style=none] (9) at (-4.5, 0) {};
		\node [style=none] (10) at (-3.5, 1) {};
		\node [style=none] (11) at (-2, 0) {};
		\node [style=none] (12) at (-7, 0.25) {};
		\node [style=black node] (13) at (-3.5, 1) {};
		\node [style=none] (14) at (0, 0) {$=$};
		\node [style=none] (15) at (3.5, 0) {};
		\node [style=none] (16) at (4, 0.25) {};
		\node [style=none] (17) at (4, -0.25) {};
		\node [style=none] (18) at (3.5, 1.5) {};
		\node [style=none] (19) at (4, 1.75) {};
		\node [style=none] (20) at (4, 1.25) {};
		\node [style=none] (21) at (6.5, 0) {};
		\node [style=none] (22) at (6, 0.25) {};
		\node [style=none] (23) at (6, -0.25) {};
		\node [style=none] (24) at (6.5, 0) {};
		\node [style=none] (25) at (6.5, 1.5) {};
		\node [style=none] (26) at (6, 1.75) {};
		\node [style=none] (27) at (6, 1.25) {};
		\node [style=none] (28) at (6.5, 1.5) {};
		\node [style=none] (29) at (2, 0) {};
		\node [style=none] (30) at (8, 0) {};
		\node [style=black node] (31) at (2.5, 1.5) {};
		\node [style=black node] (32) at (7.5, 1.5) {};
		\node [style=none] (33) at (10, 0) {$=$};
		\node [style=none] (34) at (13.5, 0) {};
		\node [style=none] (35) at (14, 0.25) {};
		\node [style=none] (36) at (14, -0.25) {};
		\node [style=none] (37) at (13.5, 1.5) {};
		\node [style=none] (38) at (14, 1.75) {};
		\node [style=none] (39) at (14, 1.25) {};
		\node [style=none] (40) at (16.5, 0) {};
		\node [style=none] (41) at (16, 0.25) {};
		\node [style=none] (42) at (16, -0.25) {};
		\node [style=none] (43) at (16.5, 0) {};
		\node [style=none] (48) at (12, 0) {};
		\node [style=none] (49) at (18, 0) {};
		\node [style=black node] (50) at (12.5, 1.5) {};
		\node [style=black node] (51) at (16, 2.25) {};
		\node [style=white node] (53) at (16, 1.25) {};
		\node [style=none] (54) at (20, 0) {$+$};
		\node [style=none] (55) at (23.5, 0) {};
		\node [style=none] (56) at (24, 0.25) {};
		\node [style=none] (57) at (24, -0.25) {};
		\node [style=none] (58) at (23.5, 1.5) {};
		\node [style=none] (59) at (24, 1.75) {};
		\node [style=none] (60) at (24, 1.25) {};
		\node [style=none] (61) at (26.5, 0) {};
		\node [style=none] (62) at (26, 0.25) {};
		\node [style=none] (63) at (26, -0.25) {};
		\node [style=none] (64) at (26.5, 0) {};
		\node [style=none] (65) at (22, 0) {};
		\node [style=none] (66) at (28, 0) {};
		\node [style=black node] (67) at (22.5, 1.5) {};
		\node [style=black node] (68) at (26, 1.25) {};
		\node [style=white node] (69) at (26, 2.25) {};
		\node [style=none] (70) at (0, -3) {$=$};
		\node [style=none] (80) at (2, -3) {};
		\node [style=none] (82) at (8, -3) {};
		\node [style=none] (86) at (10, -3) {$+$};
		\node [style=none] (101) at (13, -3) {};
		\node [style=none] (102) at (13.5, -2.75) {};
		\node [style=none] (103) at (13.5, -3.25) {};
		\node [style=none] (104) at (16.5, -3.25) {};
		\node [style=none] (105) at (12, -3) {};
		\node [style=black node] (106) at (14.25, -2.25) {};
		\node [style=none] (107) at (17, -3) {};
		\node [style=none] (108) at (16.5, -2.75) {};
		\node [style=none] (109) at (16.5, -3.25) {};
		\node [style=none] (110) at (17, -3) {};
		\node [style=black node] (111) at (15.75, -2.25) {};
		\node [style=none] (112) at (18, -3) {};
	\end{pgfonlayer}
	\begin{pgfonlayer}{edgelayer}
		\draw (1.center) to (2.center);
		\draw (2.center) to (3.center);
		\draw (3.center) to (1.center);
		\draw (0.center) to (2.center);
		\draw (3.center) to (4.center);
		\draw (7.center) to (8.center);
		\draw (8.center) to (9.center);
		\draw (9.center) to (7.center);
		\draw (8.center) to (10.center);
		\draw (9.center) to (11.center);
		\draw (12.center) to (7.center);
		\draw (15.center) to (16.center);
		\draw (16.center) to (17.center);
		\draw (17.center) to (15.center);
		\draw (18.center) to (19.center);
		\draw (19.center) to (20.center);
		\draw (20.center) to (18.center);
		\draw (21.center) to (22.center);
		\draw (22.center) to (23.center);
		\draw (23.center) to (21.center);
		\draw (25.center) to (26.center);
		\draw (26.center) to (27.center);
		\draw (27.center) to (25.center);
		\draw (20.center) to (22.center);
		\draw (16.center) to (27.center);
		\draw (19.center) to (26.center);
		\draw (17.center) to (23.center);
		\draw (29.center) to (15.center);
		\draw (24.center) to (30.center);
		\draw (31) to (18.center);
		\draw (28.center) to (32);
		\draw (34.center) to (35.center);
		\draw (35.center) to (36.center);
		\draw (36.center) to (34.center);
		\draw (37.center) to (38.center);
		\draw (38.center) to (39.center);
		\draw (39.center) to (37.center);
		\draw (40.center) to (41.center);
		\draw (41.center) to (42.center);
		\draw (42.center) to (40.center);
		\draw (39.center) to (41.center);
		\draw (36.center) to (42.center);
		\draw (48.center) to (34.center);
		\draw (43.center) to (49.center);
		\draw (50) to (37.center);
		\draw (38.center) to (51);
		\draw (35.center) to (53);
		\draw (55.center) to (56.center);
		\draw (56.center) to (57.center);
		\draw (57.center) to (55.center);
		\draw (58.center) to (59.center);
		\draw (59.center) to (60.center);
		\draw (60.center) to (58.center);
		\draw (61.center) to (62.center);
		\draw (62.center) to (63.center);
		\draw (63.center) to (61.center);
		\draw (60.center) to (62.center);
		\draw (57.center) to (63.center);
		\draw (65.center) to (55.center);
		\draw (64.center) to (66.center);
		\draw (67) to (58.center);
		\draw (56.center) to (68);
		\draw (59.center) to (69);
		\draw (80.center) to (82.center);
		\draw (101.center) to (102.center);
		\draw (102.center) to (103.center);
		\draw (103.center) to (101.center);
		\draw (103.center) to (104.center);
		\draw (105.center) to (101.center);
		\draw (102.center) to (106);
		\draw (107.center) to (108.center);
		\draw (108.center) to (109.center);
		\draw (109.center) to (107.center);
		\draw (111) to (108.center);
		\draw (110.center) to (112.center);
	\end{pgfonlayer}
\end{tikzpicture}}$$
\end{example}

\begin{example}[Hong-Ou-Mandel]
    We compute the amplitude of the beam splitter
    $\tt{BS}$ on input/output $I = (1, 1) = J$:
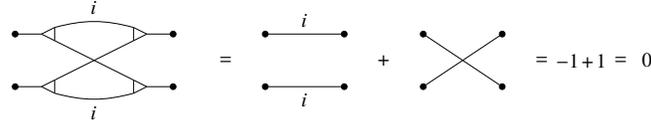
    $$\scalebox{0.7}{\begin{tikzpicture}
	\begin{pgfonlayer}{nodelayer}
		\node [style=none] (0) at (-10, 1) {};
		\node [style=none] (1) at (-9.5, 1.25) {};
		\node [style=none] (2) at (-9.5, 0.75) {};
		\node [style=none] (3) at (-6, 1) {};
		\node [style=none] (4) at (-6.5, 1.25) {};
		\node [style=none] (5) at (-6.5, 0.75) {};
		\node [style=none] (6) at (-9.5, -0.75) {};
		\node [style=none] (7) at (-6.5, -0.75) {};
		\node [style=none] (8) at (-6.5, -0.75) {};
		\node [style=none] (9) at (-6, -1) {};
		\node [style=none] (10) at (-6.5, -0.75) {};
		\node [style=none] (11) at (-6.5, -1.25) {};
		\node [style=none] (12) at (-9.5, -0.75) {};
		\node [style=none] (13) at (-10, -1) {};
		\node [style=none] (14) at (-9.5, -0.75) {};
		\node [style=none] (15) at (-9.5, -1.25) {};
		\node [style=none] (16) at (-11, -1) {};
		\node [style=none] (17) at (-9.5, -1.25) {};
		\node [style=none] (18) at (-6.5, -1.25) {};
		\node [style=none] (19) at (-11, -1) {};
		\node [style=none] (20) at (-5, -1) {};
		\node [style=none] (21) at (-11, 1) {};
		\node [style=none] (22) at (-5, 1) {};
		\node [style=black node] (23) at (-11, 1) {};
		\node [style=black node] (24) at (-5, 1) {};
		\node [style=black node] (25) at (-11, -1) {};
		\node [style=black node] (26) at (-5, -1) {};
		\node [style=none] (27) at (-8, -2) {$i$};
		\node [style=none] (29) at (-8, 2) {$i$};
		\node [style=none] (30) at (-3, 0) {$=$};
		\node [style=none] (31) at (3, 0) {$+$};
		\node [style=black node] (33) at (1.5, 1) {};
		\node [style=black node] (34) at (-1.5, 1) {};
		\node [style=black node] (35) at (1.5, -1) {};
		\node [style=black node] (36) at (-1.5, -1) {};
		\node [style=black node] (37) at (7.5, 1) {};
		\node [style=black node] (38) at (4.5, 1) {};
		\node [style=black node] (39) at (7.5, -1) {};
		\node [style=black node] (40) at (4.5, -1) {};
		\node [style=none] (41) at (0, 1.5) {$i$};
		\node [style=none] (42) at (0, -1.5) {$i$};
		\node [style=none] (43) at (9, 0) {$=$};
		\node [style=none] (44) at (12, 0) {$=$};
		\node [style=none] (45) at (13, 0) {$0$};
		\node [style=none] (47) at (10.5, 0) {$-1 + 1$};
	\end{pgfonlayer}
	\begin{pgfonlayer}{edgelayer}
		\draw (0.center) to (1.center);
		\draw (1.center) to (2.center);
		\draw (2.center) to (0.center);
		\draw (3.center) to (4.center);
		\draw (4.center) to (5.center);
		\draw (5.center) to (3.center);
		\draw [bend left=15] (1.center) to (4.center);
		\draw (2.center) to (7.center);
		\draw (5.center) to (6.center);
		\draw (9.center) to (10.center);
		\draw (10.center) to (11.center);
		\draw (11.center) to (9.center);
		\draw (13.center) to (14.center);
		\draw (14.center) to (15.center);
		\draw (15.center) to (13.center);
		\draw (16.center) to (13.center);
		\draw [in=-165, out=-15] (17.center) to (18.center);
		\draw (20.center) to (9.center);
		\draw (21.center) to (0.center);
		\draw (3.center) to (22.center);
		\draw (34) to (33);
		\draw (36) to (35);
		\draw (38) to (39);
		\draw (40) to (37);
	\end{pgfonlayer}
\end{tikzpicture}}$$
    and we recover the zero amplitude for this event.
\end{example}

We interpret a \emph{closed} diagram $d : 0 \to 0 \in \bf{QPath}$ as an event
where particle creations are matched to particle annihilations.
Given a linear optical circuit $c : m \to m \in \bf{LO}$ together with a pair of
states $I, J \in \Phi_{m, n}$ of occupation numbers, we may construct a closed diagram
$d = \bra{J} F(c) \ket{I} \in \bf{QPath}$, corresponding to the event that we
observe output $J$ when we input $I$ in a chip $c$.
Using only the $\bf{Path}$ axioms together with the normalisation rule,
we can rewrite $d$ as in the following example.
$$\scalebox{0.7}{\input{figures/path-boson-rewrite.tikz}}$$
where we use the following syntactic sugar:
$\scalebox{0.7}{\begin{tikzpicture}
	\begin{pgfonlayer}{nodelayer}
		\node [style=none] (0) at (-5, -0.75) {};
		\node [style=none] (1) at (-5, 0.75) {};
		\node [style=black node] (2) at (-6, 0) {};
		\node [style=none] (3) at (-4, 0) {:=};
		\node [style=black node] (4) at (-3, 0) {};
		\node [style=none] (5) at (-2.25, 0) {};
		\node [style=none] (6) at (-1.75, 0.25) {};
		\node [style=none] (7) at (-1.75, -0.25) {};
		\node [style=none] (8) at (-1, 0.75) {};
		\node [style=none] (9) at (-1, -0.75) {};
		\node [style=none] (10) at (5.25, 0) {};
		\node [style=none] (11) at (4.75, 0.25) {};
		\node [style=none] (12) at (4.75, -0.25) {};
		\node [style=none] (13) at (6, 0) {};
		\node [style=black node] (14) at (6, 0) {};
		\node [style=none] (15) at (4, 0.75) {};
		\node [style=none] (16) at (4, -0.75) {};
		\node [style=none] (17) at (1, -0.75) {};
		\node [style=none] (18) at (1, 0.75) {};
		\node [style=black node] (19) at (2, 0) {};
		\node [style=none] (20) at (3, 0) {$:=$};
		\node [style=none] (52) at (0, 0) {,};
	\end{pgfonlayer}
	\begin{pgfonlayer}{edgelayer}
		\draw (2) to (1.center);
		\draw (2) to (0.center);
		\draw (5.center) to (6.center);
		\draw (6.center) to (7.center);
		\draw (7.center) to (5.center);
		\draw (7.center) to (9.center);
		\draw (6.center) to (8.center);
		\draw (4) to (5.center);
		\draw (10.center) to (11.center);
		\draw (11.center) to (12.center);
		\draw (12.center) to (10.center);
		\draw (13.center) to (10.center);
		\draw (16.center) to (12.center);
		\draw (11.center) to (15.center);
		\draw (19) to (18.center);
		\draw (19) to (17.center);
	\end{pgfonlayer}
\end{tikzpicture}}$.
At the end of the rewriting process, we obtain a weighted bipartite graph.
Let us denote this graph by $G_d = (N, E)$ where $N$ is the set of nodes
and $E \sub N^2$ is the set of edges, together with $w: E \to \bb{C}$ an assignment
of complex weights to every edge. Note that $G$ is an undirected graph,
i.e. $(i, j) \in E \implies (j, i) \in E$.

\begin{proposition}[Normal form]
    Any closed diagram $d \in \bf{QPath}$ can be reduced to a pair $(G_d, N_d)$,
    where $G_d$ is weighted bipartite graph and $N_d$ is a normalisation factor,
    using the axioms in Figure \ref{path-classical} and the normalisation law.
\end{proposition}
\begin{proof}
    The normal form procedure exemplified above is the same as for $\bf{Path}$,
    with the addition of the use of the normalisation law which determines $N_d$.
\end{proof}

Once $d \in \bf{QPath}$ has been reduced to a weighted bipartite graph, we may
further reduce it down to a scalar value by using the branching and scalar laws.
Most terms obtained by branching will cancel out because of the first scalar law.
The remaining terms are found to be in one-to-one correspondence with
the \emph{perfect matchings} of $G_d$.
Recall that a matching for a graph $G$ is a subset of the edges
$M \sub E$ such that no node is contained in two edges of $M$.
A perfect matching is a matching $M$ such that every node is contained in an
edge of $M$.

\begin{theorem}[Matchings]
    For closed diagrams $d: 0 \to 0 \in \bf{QPath}$, the rewrite rules of
    $\bf{QPath}$ are complete for the bosonic intepretation
    $\B : \bf{QPath} \to \bf{Hilb}_\otimes$, which moreover satisfies:
    \begin{equation}\label{eq-perfect-matchings}
        \B(d) = N_d \sum_{M}\prod_{e \in M} w_e
    \end{equation}
    where $M$ ranges over the perfect matchings of the graph $G_d$.
\end{theorem}
\begin{proof}
    We need to show that if two closed diagrams $d, d'$ have the same intepretation
    $\B(d) = \B(d') \in \bb{C}$, then we can rewrite from $d$ to $d'$ using the
    axioms of $\bf{QPath}$.  To see this, note that for any closed diagram $d$,
    the branching law turns the graph $G_d$ into a sum of $n^n$ terms,
    where $n$ is the number of photon preparations.
    We can cancel most of these terms using the ``bone'' law, which leaves
    us with $n!$ terms corresponding to the perfect matchings of $G_d$: each
    photon preparation is matched to a photon annihilation.
    Finally we reduce each of the terms to a complex value using the scalar law.
    It is a standard result in graph theory that the sum of weights
    of perfect matchings of a graph is equal to the permanent of its
    adjacency matrix, yielding (\ref{eq-perfect-matchings}).
    Therefore we can use the axioms of $\bf{QPath}$ to reduce both $d$ and $d'$
    to the same scalar value in $\bb{C}$.
    Since all the rules of $\bf{QPath}$ are invertible we can rewrite from $d$
    to $d'$, yielding completeness. We do not currently know if the rules are
    complete also for ``open'' diagrams.
\end{proof}

\begin{example}
    For a generic event $d$ with three photons in $\bf{QPath}$, the normal form procedure
    gives us a weighted bipartite graph $G_d$ with input and output of size $3$,
    or equivalently we have a $3 \times 3$ adjacency matrix of weights.
    Using the branching law, we reduce the diagram to the following sum:
    $$\scalebox{0.7}
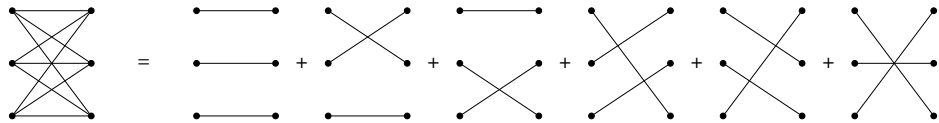
{\begin{tikzpicture}
	\begin{pgfonlayer}{nodelayer}
		\node [style=black node] (0) at (-7, 2) {};
		\node [style=black node] (1) at (-7, 0) {};
		\node [style=black node] (2) at (-7, -2) {};
		\node [style=black node] (3) at (-4, 0) {};
		\node [style=black node] (4) at (-4, 2) {};
		\node [style=black node] (5) at (-4, -2) {};
		\node [style=none] (6) at (-2, 0) {$=$};
		\node [style=black node] (7) at (0, 2) {};
		\node [style=black node] (8) at (0, 0) {};
		\node [style=black node] (9) at (0, -2) {};
		\node [style=black node] (10) at (3, 0) {};
		\node [style=black node] (11) at (3, 2) {};
		\node [style=black node] (12) at (3, -2) {};
		\node [style=black node] (13) at (5, 2) {};
		\node [style=black node] (14) at (5, 0) {};
		\node [style=black node] (15) at (5, -2) {};
		\node [style=black node] (16) at (8, 0) {};
		\node [style=black node] (17) at (8, 2) {};
		\node [style=black node] (18) at (8, -2) {};
		\node [style=black node] (19) at (15, 2) {};
		\node [style=black node] (20) at (15, 0) {};
		\node [style=black node] (21) at (15, -2) {};
		\node [style=black node] (22) at (18, 0) {};
		\node [style=black node] (23) at (18, 2) {};
		\node [style=black node] (24) at (18, -2) {};
		\node [style=black node] (25) at (25, 2) {};
		\node [style=black node] (26) at (25, 0) {};
		\node [style=black node] (27) at (25, -2) {};
		\node [style=black node] (28) at (28, 0) {};
		\node [style=black node] (29) at (28, 2) {};
		\node [style=black node] (30) at (28, -2) {};
		\node [style=black node] (31) at (10, 2) {};
		\node [style=black node] (32) at (10, 0) {};
		\node [style=black node] (33) at (10, -2) {};
		\node [style=black node] (34) at (13, 0) {};
		\node [style=black node] (35) at (13, 2) {};
		\node [style=black node] (36) at (13, -2) {};
		\node [style=black node] (37) at (20, 2) {};
		\node [style=black node] (38) at (20, 0) {};
		\node [style=black node] (39) at (20, -2) {};
		\node [style=black node] (40) at (23, 0) {};
		\node [style=black node] (41) at (23, 2) {};
		\node [style=black node] (42) at (23, -2) {};
		\node [style=none] (43) at (4, 0) {$+$};
		\node [style=none] (44) at (9, 0) {$+$};
		\node [style=none] (45) at (19, 0) {$+$};
		\node [style=none] (46) at (24, 0) {$+$};
		\node [style=none] (47) at (14, 0) {$+$};
	\end{pgfonlayer}
	\begin{pgfonlayer}{edgelayer}
		\draw (0) to (3);
		\draw (0) to (4);
		\draw (0) to (5);
		\draw (1) to (4);
		\draw (1) to (3);
		\draw (2) to (3);
		\draw (1) to (5);
		\draw (2) to (5);
		\draw (2) to (4);
		\draw (7) to (11);
		\draw (8) to (10);
		\draw (9) to (12);
		\draw (13) to (16);
		\draw (14) to (17);
		\draw (15) to (18);
		\draw (19) to (24);
		\draw (20) to (23);
		\draw (21) to (22);
		\draw (25) to (30);
		\draw (26) to (28);
		\draw (27) to (29);
		\draw (31) to (35);
		\draw (33) to (34);
		\draw (32) to (36);
		\draw (37) to (40);
		\draw (38) to (42);
		\draw (39) to (41);
	\end{pgfonlayer}
\end{tikzpicture}}$$
    equivalently, we have just split the graph into its perfect matchings.
    Now we use the scalar laws to reduce each term to a complex number. Equivalently,
    we multiply the weights assigned to the edges on each matching. Finally we sum those
    terms to obtain the amplitude. Equivalently, we have computed the permanent
    of the adjacency matrix of $G_d$.
\end{example}

\section{Linear-optical quantum computing}\label{sec-computing}
%!TEX root = ms.tex

Our aim in this section is to describe how linear optics is used for qubit quantum
computation. We will do this by giving a complete mapping from the
$\bf{ZX}$ calculus to $\bf{QPath}$.
We start by introducing the ZX calculus on qubits. The dual-rail encoding allows to
encode a logical qubit as a photon in a pair of spatial modes. We show how all
single qubit unitaries may be applied using simple linear optical devices.
We describe fusion measurements as diagrams in $\bf{QPath}$ and show
how they can be used, along with polarising beam splitters, to construct Bell
states and more general cluster states.

\paragraph{ZX calculus.}

The ZX calculus is a graphical language for reasoning about qubit quantum
computation. It has strong links with both circuit-based and measurement-based
models of quantum computing \cite{backens2021}.
The $\bf{ZX}$ calculus is generated by the following basic operations:
$$\scalebox{0.8}
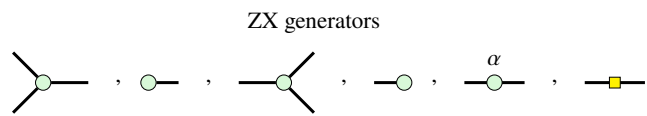
{\begin{tikzpicture}
	\begin{pgfonlayer}{nodelayer}
		\node [style=none] (0) at (0, 2) {ZX generators};
		\node [style=green node] (1) at (-5.5, 0) {};
		\node [style=none] (2) at (-4.5, 0) {};
		\node [style=none] (3) at (5, 0) {};
		\node [style=none] (4) at (7, 0) {};
		\node [style=green node] (5) at (6, 0) {};
		\node [style=none] (6) at (9, 0) {};
		\node [style=none] (7) at (11, 0) {};
		\node [style=none] (8) at (-10, 1) {};
		\node [style=none] (9) at (-10, -1) {};
		\node [style=none] (10) at (-7.5, 0) {};
		\node [style=none] (11) at (-2.5, 0) {};
		\node [style=none] (12) at (0, 1) {};
		\node [style=none] (13) at (0, -1) {};
		\node [style=green node] (14) at (-9, 0) {};
		\node [style=green node] (15) at (-1, 0) {};
		\node [style=none] (16) at (-6.5, 0) {,};
		\node [style=none] (17) at (8, 0) {,};
		\node [style=none] (18) at (4, 0) {,};
		\node [style=none] (19) at (-3.5, 0) {,};
		\node [style=none] (20) at (6, 0.75) {$\alpha$};
		\node [style=hadamard] (21) at (10, 0) {};
		\node [style=green node] (22) at (3, 0) {};
		\node [style=none] (23) at (2, 0) {};
		\node [style=none] (24) at (1, 0) {,};
	\end{pgfonlayer}
	\begin{pgfonlayer}{edgelayer}
		\draw [style=thick] (1) to (2.center);
		\draw [style=thick] (3.center) to (4.center);
		\draw [style=thick] (6.center) to (7.center);
		\draw [style=thick] (8.center) to (14);
		\draw [style=thick] (14) to (9.center);
		\draw [style=thick] (14) to (10.center);
		\draw [style=thick] (11.center) to (15);
		\draw [style=thick] (15) to (12.center);
		\draw [style=thick] (15) to (13.center);
		\draw [style=thick] (22) to (23.center);
	\end{pgfonlayer}
\end{tikzpicture}}$$
We will also use the following syntactic sugar for $X$ states and phases:
$$\scalebox{0.8}{\begin{tikzpicture}
	\begin{pgfonlayer}{nodelayer}
		\node [style=red node] (0) at (2.5, 0) {};
		\node [style=none] (1) at (5.5, 0) {};
		\node [style=none] (2) at (8.5, 0) {};
		\node [style=green node] (3) at (7, 0) {};
		\node [style=none] (4) at (7, 0.75) {$\alpha$};
		\node [style=none] (5) at (1.5, 0) {};
		\node [style=none] (6) at (3.5, 0) {};
		\node [style=none] (7) at (2.5, 0.75) {$\alpha$};
		\node [style=hadamard] (8) at (6.25, 0) {};
		\node [style=hadamard] (9) at (7.75, 0) {};
		\node [style=none] (10) at (4.5, 0) {:=};
		\node [style=none] (11) at (0, 0) {,};
		\node [style=red node] (12) at (-6.5, 0) {};
		\node [style=green node] (13) at (-3.5, 0) {};
		\node [style=hadamard] (14) at (-2.5, 0) {};
		\node [style=none] (15) at (-5.5, 0) {};
		\node [style=none] (16) at (-1.5, 0) {};
		\node [style=none] (17) at (-4.5, 0) {$:=$};
		\node [style=none] (18) at (3, 0) {};
	\end{pgfonlayer}
	\begin{pgfonlayer}{edgelayer}
		\draw [style=thick] (1.center) to (2.center);
		\draw [style=thick] (5.center) to (6.center);
		\draw [style=thick] (12) to (15.center);
		\draw [style=thick] (13) to (16.center);
	\end{pgfonlayer}
\end{tikzpicture}}$$
In this section we are not really interested in the rewrite rules for ZX diagrams,
rather in their interpretation as linear maps between qubits. We will give this
interpretation as we map each $\bf{ZX}$ generator to post-selected optical circuits
in $\bf{QPath}$, and refer to \cite{vandewetering2020, CKbook} for more
in-depth discussions.

\paragraph{Dual rail qubits.}

The dual-rail encoding can be thought of as a translation between polarized and
spatial modes of photons. The polarization states of a single photon are spanned by the basis
states $\ket{H}, \ket{V}$ for horizontal and vertical polarization, and thus
naturally form a \emph{qubit}.
The dual-rail encoding consists in encoding a polarised mode of light as a
pair of spatial modes in $\bf{LO}$ under the mapping
$\ket{H} \mapsto \ket{0, 1}, \, \ket{V} \mapsto \ket{1, 0}$.
The $Z$ basis of a dual rail qubit may be expressed as a pair of $\bf{QPath}$ diagrams:
$$\scalebox{0.9}{\begin{tikzpicture}
	\begin{pgfonlayer}{nodelayer}
		\node [style=none] (0) at (-1.5, 0.5) {};
		\node [style=none] (1) at (-2.5, 0.5) {};
		\node [style=none] (2) at (-1.5, -0.5) {};
		\node [style=none] (3) at (-2.5, -0.5) {};
		\node [style=black node] (4) at (-2.5, -0.5) {};
		\node [style=white node] (5) at (-2.5, 0.5) {};
		\node [style=none] (6) at (2.5, -0.5) {};
		\node [style=none] (7) at (1.5, -0.5) {};
		\node [style=none] (8) at (2.5, 0.5) {};
		\node [style=none] (9) at (1.5, 0.5) {};
		\node [style=black node] (10) at (1.5, 0.5) {};
		\node [style=white node] (11) at (1.5, -0.5) {};
		\node [style=none] (12) at (0, 0) {,};
		\node [style=none] (14) at (9, 0) {$\set{\ket{H}, \ket{V}}$};
		\node [style=red node] (52) at (-11, 0) {};
		\node [style=red node] (53) at (-7.5, 0) {};
		\node [style=none] (54) at (-9.5, 0) {};
		\node [style=none] (55) at (-6, 0) {};
		\node [style=none] (56) at (-8.5, 0) {$,$};
		\node [style=none] (57) at (-7.5, 0.75) {$\pi$};
	\end{pgfonlayer}
	\begin{pgfonlayer}{edgelayer}
		\draw (1.center) to (0.center);
		\draw (3.center) to (2.center);
		\draw (7.center) to (6.center);
		\draw (9.center) to (8.center);
		\draw [style=thick] (52) to (54.center);
		\draw [style=thick] (53) to (55.center);
	\end{pgfonlayer}
\end{tikzpicture}}$$
The $Z$ and $X$ effects correspond to the following diagrams:
$$\scalebox{0.8}{\begin{tikzpicture}
	\begin{pgfonlayer}{nodelayer}
		\node [style=none] (3) at (-9.25, 0) {};
		\node [style=green node] (4) at (-7.75, 0) {};
		\node [style=black node] (5) at (3, 0) {};
		\node [style=none] (12) at (1, -1) {};
		\node [style=none] (13) at (1, 1) {};
		\node [style=red node] (15) at (-11.5, 0) {};
		\node [style=none] (16) at (-13, 0) {};
		\node [style=black node] (19) at (-1.5, -0.75) {};
		\node [style=white node] (20) at (-1.5, 0.75) {};
		\node [style=none] (21) at (11, 0) {$\{\bra{H}\, , \, \bra{H} + \bra{V}\}$};
		\node [style=none] (22) at (-10.5, 0) {,};
		\node [style=none] (23) at (0, 0) {,};
		\node [style=none] (24) at (-3, 0.75) {};
		\node [style=none] (25) at (-3, -0.75) {};
	\end{pgfonlayer}
	\begin{pgfonlayer}{edgelayer}
		\draw [style=thick] (3.center) to (4);
		\draw (5) to (13.center);
		\draw (12.center) to (5);
		\draw [style=thick] (16.center) to (15);
		\draw (24.center) to (20);
		\draw (25.center) to (19);
	\end{pgfonlayer}
\end{tikzpicture}}$$
The $Z$ effect may be implemented by post-selecting a photon detector,
the $X$ effect by precomposition with a beam splitter.
$Z$ phases on dual-rail qubits are obtained as follows:
$$\scalebox{0.9}{\begin{tikzpicture}
	\begin{pgfonlayer}{nodelayer}
		\node [style=none] (4) at (-1.5, 0.5) {};
		\node [style=none] (5) at (-1.5, -0.5) {};
		\node [style=none] (6) at (1.5, 0.5) {};
		\node [style=none] (7) at (1.5, -0.5) {};
		\node [style=none] (9) at (-5.5, 0.75) {$\alpha$};
		\node [style=none] (10) at (0, -1) {$e^{i\alpha}$};
		\node [style=none] (14) at (6, 0.5) {};
		\node [style=none] (15) at (6, 0.5) {$\ket{H} \, \mapsto \, \ket{H}$};
		\node [style=none] (16) at (6, -0.5) {$\ket{V} \, \mapsto \, e^{i \alpha} \ket{V}$};
		\node [style=none] (17) at (-7, 0) {};
		\node [style=none] (18) at (-4, 0) {};
		\node [style=green node] (19) at (-5.5, 0) {};
	\end{pgfonlayer}
	\begin{pgfonlayer}{edgelayer}
		\draw (4.center) to (6.center);
		\draw (5.center) to (7.center);
		\draw [style=thick] (17.center) to (18.center);
	\end{pgfonlayer}
\end{tikzpicture}}$$
The rotations from the $Z$ basis to the $X$ and $Y$ bases are given by the beam splitters $\tt{BS}_H$ and $\tt{BS}$ respectively, defined as follows:
$$\scalebox{0.9}{\begin{tikzpicture}
	\begin{pgfonlayer}{nodelayer}
		\node [style=none] (11) at (10.5, 0.75) {$\ket{H} \, \mapsto \, \ket{H} + \ket{V}$};
		\node [style=none] (12) at (10.5, -0.75) {$\ket{V} \, \mapsto \, \ket{H} - \ket{V}$};
		\node [style=none] (13) at (-3.5, 0.25) {};
		\node [style=none] (14) at (-3.5, 0) {};
		\node [style=none] (15) at (-2.5, 0) {};
		\node [style=none] (16) at (-2.5, 0.25) {};
		\node [style=none] (17) at (-3, 0.25) {};
		\node [style=none] (18) at (-3, 0) {};
		\node [style=none] (19) at (-1.5, 1) {};
		\node [style=none] (20) at (-1.5, -0.75) {};
		\node [style=none] (21) at (-4.5, 1) {};
		\node [style=none] (22) at (-4.5, -0.75) {};
		\node [style=none] (23) at (-3, -0.75) {$H$};
		\node [style=none] (24) at (0, 0) {$=$};
		\node [style=none] (25) at (2.5, 0.75) {};
		\node [style=none] (26) at (3, 1) {};
		\node [style=none] (27) at (3, 0.5) {};
		\node [style=none] (28) at (5, 0.75) {};
		\node [style=none] (29) at (4.5, 1) {};
		\node [style=none] (30) at (4.5, 0.5) {};
		\node [style=none] (31) at (3, -0.5) {};
		\node [style=none] (32) at (4.5, -0.5) {};
		\node [style=none] (33) at (4.5, -0.5) {};
		\node [style=none] (34) at (5, -0.75) {};
		\node [style=none] (35) at (4.5, -0.5) {};
		\node [style=none] (36) at (4.5, -1) {};
		\node [style=none] (37) at (3, -0.5) {};
		\node [style=none] (38) at (2.5, -0.75) {};
		\node [style=none] (39) at (3, -0.5) {};
		\node [style=none] (40) at (3, -1) {};
		\node [style=none] (41) at (3, -1) {};
		\node [style=none] (42) at (4.5, -1) {};
		\node [style=none] (44) at (1.5, 0.75) {};
		\node [style=none] (45) at (1.5, -0.75) {};
		\node [style=none] (46) at (6, 0.75) {};
		\node [style=none] (47) at (6, -0.75) {};
		\node [style=none] (49) at (10.5, -3.25) {$\ket{H} \, \mapsto \, i\ket{H} + \ket{V}$};
		\node [style=none] (50) at (10.5, -4.75) {$\ket{V} \, \mapsto \, \ket{H} + i\ket{V}$};
		\node [style=none] (51) at (-3.5, -3.75) {};
		\node [style=none] (52) at (-3.5, -4) {};
		\node [style=none] (53) at (-2.5, -4) {};
		\node [style=none] (54) at (-2.5, -3.75) {};
		\node [style=none] (55) at (-3, -3.75) {};
		\node [style=none] (56) at (-3, -4) {};
		\node [style=none] (57) at (-1.5, -3) {};
		\node [style=none] (58) at (-1.5, -4.75) {};
		\node [style=none] (59) at (-4.5, -3) {};
		\node [style=none] (60) at (-4.5, -4.75) {};
		\node [style=none] (62) at (0, -4) {$=$};
		\node [style=none] (63) at (2.5, -3.25) {};
		\node [style=none] (64) at (3, -3) {};
		\node [style=none] (65) at (3, -3.5) {};
		\node [style=none] (66) at (5, -3.25) {};
		\node [style=none] (67) at (4.5, -3) {};
		\node [style=none] (68) at (4.5, -3.5) {};
		\node [style=none] (69) at (3, -4.5) {};
		\node [style=none] (70) at (4.5, -4.5) {};
		\node [style=none] (71) at (4.5, -4.5) {};
		\node [style=none] (72) at (5, -4.75) {};
		\node [style=none] (73) at (4.5, -4.5) {};
		\node [style=none] (74) at (4.5, -5) {};
		\node [style=none] (75) at (3, -4.5) {};
		\node [style=none] (76) at (2.5, -4.75) {};
		\node [style=none] (77) at (3, -4.5) {};
		\node [style=none] (78) at (3, -5) {};
		\node [style=none] (79) at (3, -5) {};
		\node [style=none] (80) at (4.5, -5) {};
		\node [style=none] (82) at (1.5, -3.25) {};
		\node [style=none] (83) at (1.5, -4.75) {};
		\node [style=none] (84) at (6, -3.25) {};
		\node [style=none] (85) at (6, -4.75) {};
		\node [style=hadamard] (86) at (-9, 0) {};
		\node [style=none] (87) at (-11, 0) {};
		\node [style=none] (88) at (-7, 0) {};
		\node [style=none] (90) at (-11, -3.75) {};
		\node [style=none] (91) at (-7, -3.75) {};
		\node [style=red node] (92) at (-9, -3.75) {};
		\node [style=none] (93) at (-9, -2.75) {$\frac{\pi}{2}$};
	\end{pgfonlayer}
	\begin{pgfonlayer}{edgelayer}
		\draw (13.center) to (16.center);
		\draw (16.center) to (15.center);
		\draw (15.center) to (14.center);
		\draw (14.center) to (13.center);
		\draw [bend left=15] (21.center) to (17.center);
		\draw [bend left=15] (17.center) to (19.center);
		\draw [bend right=15] (18.center) to (20.center);
		\draw [bend left=15] (18.center) to (22.center);
		\draw (25.center) to (26.center);
		\draw (26.center) to (27.center);
		\draw (27.center) to (25.center);
		\draw (28.center) to (29.center);
		\draw (29.center) to (30.center);
		\draw (30.center) to (28.center);
		\draw (27.center) to (32.center);
		\draw (30.center) to (31.center);
		\draw (34.center) to (35.center);
		\draw (35.center) to (36.center);
		\draw (36.center) to (34.center);
		\draw (38.center) to (39.center);
		\draw (39.center) to (40.center);
		\draw (40.center) to (38.center);
		\draw (28.center) to (46.center);
		\draw (34.center) to (47.center);
		\draw (44.center) to (25.center);
		\draw (45.center) to (38.center);
		\draw (51.center) to (54.center);
		\draw (54.center) to (53.center);
		\draw (53.center) to (52.center);
		\draw (52.center) to (51.center);
		\draw [bend left=15] (59.center) to (55.center);
		\draw [bend left=15] (55.center) to (57.center);
		\draw [bend right=15] (56.center) to (58.center);
		\draw [bend left=15] (56.center) to (60.center);
		\draw (63.center) to (64.center);
		\draw (64.center) to (65.center);
		\draw (65.center) to (63.center);
		\draw (66.center) to (67.center);
		\draw (67.center) to (68.center);
		\draw (68.center) to (66.center);
		\draw (65.center) to (70.center);
		\draw (68.center) to (69.center);
		\draw (72.center) to (73.center);
		\draw (73.center) to (74.center);
		\draw (74.center) to (72.center);
		\draw (76.center) to (77.center);
		\draw (77.center) to (78.center);
		\draw (78.center) to (76.center);
		\draw (66.center) to (84.center);
		\draw (72.center) to (85.center);
		\draw (82.center) to (63.center);
		\draw (83.center) to (76.center);
		\draw (26.center) to (29.center);
		\draw [style=red-edge] (41.center) to (42.center);
		\draw [style=blue-edge] (64.center) to (67.center);
		\draw [style=blue-edge] (79.center) to (80.center);
		\draw [style=thick] (87.center) to (86);
		\draw [style=thick] (88.center) to (86);
		\draw [style=thick] (90.center) to (91.center);
	\end{pgfonlayer}
\end{tikzpicture}}$$
where we use the following syntactic sugar:
$\scalebox{0.9}{\begin{tikzpicture}
	\begin{pgfonlayer}{nodelayer}
		\node [style=none] (30) at (-7, 0) {};
		\node [style=none] (31) at (-5, 0) {};
		\node [style=none] (32) at (-3, 0) {};
		\node [style=none] (33) at (-1, 0) {};
		\node [style=none] (34) at (1, 0) {};
		\node [style=none] (35) at (3, 0) {};
		\node [style=none] (36) at (-4, 0) {:=};
		\node [style=none] (37) at (-2, 0.5) {$-1$};
		\node [style=none] (38) at (5.25, 0) {};
		\node [style=none] (39) at (7.25, 0) {};
		\node [style=none] (40) at (4, 0) {:=};
		\node [style=none] (41) at (6.25, 0.5) {$i$};
		\node [style=none] (42) at (0, 0) {,};
	\end{pgfonlayer}
	\begin{pgfonlayer}{edgelayer}
		\draw [style=red-edge] (30.center) to (31.center);
		\draw [style=blue-edge] (34.center) to (35.center);
		\draw (32.center) to (33.center);
		\draw (38.center) to (39.center);
	\end{pgfonlayer}
\end{tikzpicture}}$.
We give the encoding up to scalar factor which does not affect the logic of the mapping. For example, the hadamard gate is technically $\frac{1}{\sqrt{2}} \tt{BS}_H$. In conjunction with $Z$ phases, we can use it to obtain all single
qubit unitaries in dual rail encoding.

\begin{example}[HOM]\label{ex-hom}
    Using the beam splitters above, we obtain two versions of the Hong-Ou-Mandel
    effect which are depicted graphically as follows:
    $$\scalebox{0.8}{\begin{tikzpicture}
	\begin{pgfonlayer}{nodelayer}
		\node [style=new style 0] (0) at (-6, 1) {};
		\node [style=new style 0] (1) at (-3, 1) {};
		\node [style=new style 0] (2) at (-3, -1) {};
		\node [style=new style 0] (3) at (-6, -1) {};
		\node [style=none] (4) at (-1.5, 0) {$=$};
		\node [style=none] (5) at (0, 0) {$0$};
		\node [style=none] (6) at (1.5, 0) {$=$};
		\node [style=black node] (7) at (3, 1) {};
		\node [style=black node] (8) at (3, -1) {};
		\node [style=black node] (9) at (6, -1) {};
		\node [style=black node] (10) at (6, 1) {};
	\end{pgfonlayer}
	\begin{pgfonlayer}{edgelayer}
		\draw [style=blue-edge] (0) to (1);
		\draw [style=blue-edge] (3) to (2);
		\draw (0) to (2);
		\draw (3) to (1);
		\draw [style=red-edge] (8) to (9);
		\draw (7) to (9);
		\draw (8) to (10);
		\draw (7) to (10);
	\end{pgfonlayer}
\end{tikzpicture}}$$
\end{example}

\paragraph{Fusion measurements.}

Fusion measurements are Bell measurements on dual-rail qubits.
They correspond to the linear map:
$ \ket{H,H} \, \mapsto \, \ket{H} \, , \, \ket{V,V} \, \mapsto \, \ket{V} \, ,
\, \ket{H,V}, \ket{V,H} \, \mapsto \, 0$, which is denoted as a green spider
with two inputs and one output in $\bf{ZX}$,
and is obtained on dual-rail qubits as the following diagram in $\bf{QPath}$.
$$\scalebox{0.8}
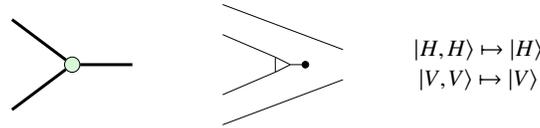
{\begin{tikzpicture}
	\begin{pgfonlayer}{nodelayer}
		\node [style=none] (1) at (-9, 1.5) {};
		\node [style=none] (2) at (-9, -1.5) {};
		\node [style=none] (3) at (-5, 0) {};
		\node [style=none] (5) at (-2, 2) {};
		\node [style=none] (7) at (-2, -1) {};
		\node [style=none] (8) at (-2, -2) {};
		\node [style=none] (9) at (2, 0.5) {};
		\node [style=none] (10) at (2, -0.5) {};
		\node [style=none] (11) at (0.25, 0) {};
		\node [style=none] (12) at (-0.25, 0.25) {};
		\node [style=none] (13) at (-0.25, -0.25) {};
		\node [style=none] (14) at (0.75, 0) {};
		\node [style=black node] (15) at (0.75, 0) {};
		\node [style=none] (16) at (-2, 1) {};
		\node [style=none] (17) at (-2, -1) {};
		\node [style=green node] (18) at (-7, 0) {};
		\node [style=none] (19) at (6.5, 0.5) {$\ket{H,H} \mapsto\ket{H}$};
		\node [style=none] (20) at (6.5, -0.5) {$\ket{V,V} \mapsto\ket{V}$};
	\end{pgfonlayer}
	\begin{pgfonlayer}{edgelayer}
		\draw (11.center) to (12.center);
		\draw (12.center) to (13.center);
		\draw (13.center) to (11.center);
		\draw (11.center) to (14.center);
		\draw (16.center) to (12.center);
		\draw (17.center) to (13.center);
		\draw (5.center) to (9.center);
		\draw (8.center) to (10.center);
		\draw [style=thick] (1.center) to (18);
		\draw [style=thick] (18) to (2.center);
		\draw [style=thick] (18) to (3.center);
	\end{pgfonlayer}
\end{tikzpicture}}$$
To see that this measures the Bell basis, note that there must be exactly one photon
in the two middle modes. The input basis state in dual-rail encoding are
$\set{\ket{0101}, \ket{0110}, \ket{1010}, \ket{1001}}$ and this condition
is satisfied only by $\ket{1010}$ and $\ket{0101}$ which correspond respectively
to $\ket{H,H}$ and $\ket{V,V}$.

\paragraph{Bell states.}

We engineer a representation of dual-rail bell states as $\bf{QPath}$ diagrams.
$$\scalebox{0.8}
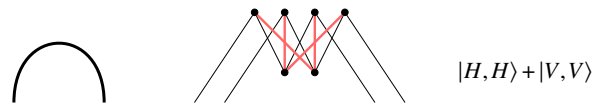
{\begin{tikzpicture}
	\begin{pgfonlayer}{nodelayer}
		\node [style=black node] (0) at (1, 4) {};
		\node [style=black node] (1) at (-1, 4) {};
		\node [style=black node] (2) at (0, 4) {};
		\node [style=black node] (3) at (2, 4) {};
		\node [style=new style 0] (4) at (1, 2) {};
		\node [style=new style 0] (5) at (0, 2) {};
		\node [style=none] (6) at (-3, 1) {};
		\node [style=none] (7) at (-2, 1) {};
		\node [style=none] (8) at (3, 1) {};
		\node [style=none] (9) at (4, 1) {};
		\node [style=none] (20) at (-9, 1) {};
		\node [style=none] (21) at (-6, 1) {};
		\node [style=none] (26) at (8, 2) {$\ket{H,H} + \ket{V,V}$};
	\end{pgfonlayer}
	\begin{pgfonlayer}{edgelayer}
		\draw (3) to (4);
		\draw (0) to (5);
		\draw (1) to (5);
		\draw (2) to (4);
		\draw (3) to (9.center);
		\draw (0) to (8.center);
		\draw (1) to (6.center);
		\draw (2) to (7.center);
		\draw [style=blue-edge] (3) to (5);
		\draw [style=blue-edge] (0) to (4);
		\draw [style=blue-edge] (1) to (4);
		\draw [style=blue-edge] (2) to (5);
		\draw [style=thick, bend left=90, looseness=2.25] (20.center) to (21.center);
	\end{pgfonlayer}
\end{tikzpicture}}$$
We can check that this diagram corresponds to the bell state
$\ket{H,H} + \ket{V,V}$ by branching:
$$\scalebox{0.8}
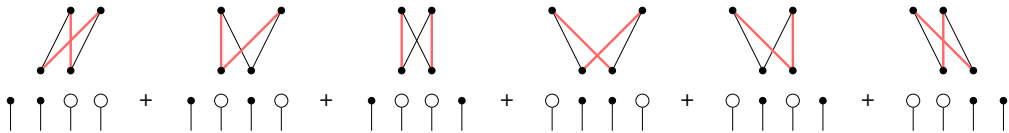
{\begin{tikzpicture}
	\begin{pgfonlayer}{nodelayer}
		\node [style=black node] (0) at (-7, 1) {};
		\node [style=black node] (1) at (-6, 1) {};
		\node [style=black node] (2) at (-1, 1) {};
		\node [style=black node] (3) at (1, 1) {};
		\node [style=black node] (4) at (5, 1) {};
		\node [style=black node] (5) at (8, 1) {};
		\node [style=white node] (6) at (-4, 1) {};
		\node [style=white node] (7) at (-5, 1) {};
		\node [style=white node] (8) at (0, 1) {};
		\node [style=white node] (9) at (2, 1) {};
		\node [style=white node] (10) at (6, 1) {};
		\node [style=white node] (11) at (7, 1) {};
		\node [style=none] (12) at (-7, 0) {};
		\node [style=none] (13) at (-6, 0) {};
		\node [style=none] (14) at (-5, 0) {};
		\node [style=none] (15) at (-4, 0) {};
		\node [style=none] (16) at (-1, 0) {};
		\node [style=none] (17) at (0, 0) {};
		\node [style=none] (18) at (1, 0) {};
		\node [style=none] (19) at (2, 0) {};
		\node [style=none] (20) at (5, 0) {};
		\node [style=none] (21) at (6, 0) {};
		\node [style=none] (22) at (7, 0) {};
		\node [style=none] (23) at (8, 0) {};
		\node [style=none] (24) at (-2.5, 1) {$+$};
		\node [style=none] (25) at (3.5, 1) {$+$};
		\node [style=none] (26) at (9.5, 1) {$+$};
		\node [style=black node] (27) at (12, 1) {};
		\node [style=black node] (28) at (13, 1) {};
		\node [style=black node] (29) at (18, 1) {};
		\node [style=black node] (30) at (20, 1) {};
		\node [style=black node] (31) at (25, 1) {};
		\node [style=black node] (32) at (26, 1) {};
		\node [style=white node] (33) at (11, 1) {};
		\node [style=white node] (34) at (14, 1) {};
		\node [style=white node] (35) at (23, 1) {};
		\node [style=white node] (36) at (24, 1) {};
		\node [style=white node] (37) at (19, 1) {};
		\node [style=white node] (38) at (17, 1) {};
		\node [style=none] (39) at (11, 0) {};
		\node [style=none] (40) at (12, 0) {};
		\node [style=none] (41) at (13, 0) {};
		\node [style=none] (42) at (14, 0) {};
		\node [style=none] (43) at (17, 0) {};
		\node [style=none] (44) at (18, 0) {};
		\node [style=none] (45) at (19, 0) {};
		\node [style=none] (46) at (20, 0) {};
		\node [style=none] (47) at (23, 0) {};
		\node [style=none] (48) at (24, 0) {};
		\node [style=none] (49) at (25, 0) {};
		\node [style=none] (50) at (26, 0) {};
		\node [style=none] (51) at (21.5, 1) {$+$};
		\node [style=none] (52) at (15.5, 1) {$+$};
		\node [style=black node] (53) at (-5, 4) {};
		\node [style=black node] (56) at (-4, 4) {};
		\node [style=new style 0] (57) at (-5, 2) {};
		\node [style=new style 0] (58) at (-6, 2) {};
		\node [style=black node] (61) at (0, 4) {};
		\node [style=black node] (62) at (2, 4) {};
		\node [style=new style 0] (63) at (1, 2) {};
		\node [style=new style 0] (64) at (0, 2) {};
		\node [style=black node] (65) at (7, 4) {};
		\node [style=black node] (67) at (6, 4) {};
		\node [style=new style 0] (69) at (7, 2) {};
		\node [style=new style 0] (70) at (6, 2) {};
		\node [style=black node] (72) at (11, 4) {};
		\node [style=black node] (74) at (14, 4) {};
		\node [style=new style 0] (75) at (13, 2) {};
		\node [style=new style 0] (76) at (12, 2) {};
		\node [style=black node] (77) at (19, 4) {};
		\node [style=black node] (78) at (17, 4) {};
		\node [style=new style 0] (81) at (19, 2) {};
		\node [style=new style 0] (82) at (18, 2) {};
		\node [style=black node] (84) at (23, 4) {};
		\node [style=black node] (85) at (24, 4) {};
		\node [style=new style 0] (87) at (25, 2) {};
		\node [style=new style 0] (88) at (24, 2) {};
	\end{pgfonlayer}
	\begin{pgfonlayer}{edgelayer}
		\draw (0) to (12.center);
		\draw (1) to (13.center);
		\draw (7) to (14.center);
		\draw (6) to (15.center);
		\draw (2) to (16.center);
		\draw (8) to (17.center);
		\draw (3) to (18.center);
		\draw (9) to (19.center);
		\draw (4) to (20.center);
		\draw (10) to (21.center);
		\draw (11) to (22.center);
		\draw (5) to (23.center);
		\draw (33) to (39.center);
		\draw (27) to (40.center);
		\draw (28) to (41.center);
		\draw (34) to (42.center);
		\draw (38) to (43.center);
		\draw (29) to (44.center);
		\draw (37) to (45.center);
		\draw (30) to (46.center);
		\draw (35) to (47.center);
		\draw (36) to (48.center);
		\draw (31) to (49.center);
		\draw (32) to (50.center);
		\draw (56) to (57);
		\draw (53) to (58);
		\draw [style=blue-edge] (56) to (58);
		\draw [style=blue-edge] (53) to (57);
		\draw (62) to (63);
		\draw (61) to (63);
		\draw [style=blue-edge] (62) to (64);
		\draw [style=blue-edge] (61) to (64);
		\draw (65) to (70);
		\draw (67) to (69);
		\draw [style=blue-edge] (65) to (69);
		\draw [style=blue-edge] (67) to (70);
		\draw (74) to (75);
		\draw (72) to (76);
		\draw [style=blue-edge] (74) to (76);
		\draw [style=blue-edge] (72) to (75);
		\draw (77) to (82);
		\draw (78) to (82);
		\draw [style=blue-edge] (77) to (81);
		\draw [style=blue-edge] (78) to (81);
		\draw (84) to (88);
		\draw (85) to (87);
		\draw [style=blue-edge] (84) to (87);
		\draw [style=blue-edge] (85) to (88);
	\end{pgfonlayer}
\end{tikzpicture}}$$
and using the Hong-Ou-Mandel effect (Example \ref{ex-hom}).
Similarly, the Bell state $\ket{HH} - \ket{VV}$ may be represented using blue edes as follows:
$$\scalebox{0.8}
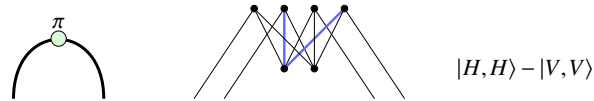
{\begin{tikzpicture}
	\begin{pgfonlayer}{nodelayer}
		\node [style=black node] (0) at (0.5, 2) {};
		\node [style=black node] (1) at (-1.5, 2) {};
		\node [style=black node] (2) at (-0.5, 2) {};
		\node [style=black node] (3) at (1.5, 2) {};
		\node [style=new style 0] (4) at (0.5, 0) {};
		\node [style=new style 0] (5) at (-0.5, 0) {};
		\node [style=none] (6) at (-3.5, -1) {};
		\node [style=none] (7) at (-2.5, -1) {};
		\node [style=none] (8) at (2.5, -1) {};
		\node [style=none] (9) at (3.5, -1) {};
		\node [style=none] (10) at (-9.5, -1) {};
		\node [style=none] (11) at (-6.5, -1) {};
		\node [style=green node] (12) at (-8, 1) {};
		\node [style=none] (13) at (-8, 1.5) {$\pi$};
		\node [style=none] (14) at (7.5, 0) {$\ket{H,H} - \ket{V,V}$};
	\end{pgfonlayer}
	\begin{pgfonlayer}{edgelayer}
		\draw (0) to (5);
		\draw (1) to (5);
		\draw (2) to (4);
		\draw (3) to (9.center);
		\draw (0) to (8.center);
		\draw (1) to (6.center);
		\draw (2) to (7.center);
		\draw [style=red-edge] (3) to (5);
		\draw [style=red-edge] (2) to (5);
		\draw (1) to (4);
		\draw (3) to (4);
		\draw (0) to (4);
		\draw [style=thick, bend left=90, looseness=2.25] (10.center) to (11.center);
	\end{pgfonlayer}
\end{tikzpicture}}$$
and we can check this using branching and HOM:
$$\scalebox{0.8}
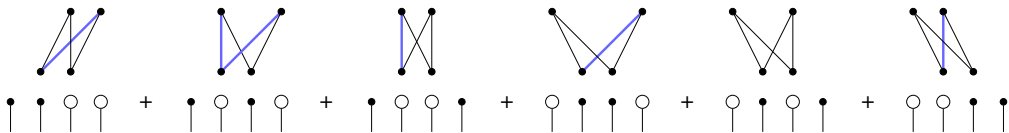
{\begin{tikzpicture}
	\begin{pgfonlayer}{nodelayer}
		\node [style=black node] (6) at (-6, 1) {};
		\node [style=black node] (7) at (-5, 1) {};
		\node [style=black node] (8) at (0, 1) {};
		\node [style=black node] (9) at (2, 1) {};
		\node [style=black node] (10) at (6, 1) {};
		\node [style=black node] (11) at (9, 1) {};
		\node [style=black node] (12) at (13, 1) {};
		\node [style=black node] (13) at (14, 1) {};
		\node [style=black node] (14) at (19, 1) {};
		\node [style=black node] (15) at (21, 1) {};
		\node [style=black node] (16) at (26, 1) {};
		\node [style=black node] (17) at (27, 1) {};
		\node [style=white node] (18) at (12, 1) {};
		\node [style=white node] (19) at (15, 1) {};
		\node [style=white node] (20) at (-3, 1) {};
		\node [style=white node] (21) at (-4, 1) {};
		\node [style=white node] (22) at (1, 1) {};
		\node [style=white node] (23) at (3, 1) {};
		\node [style=white node] (24) at (7, 1) {};
		\node [style=white node] (25) at (8, 1) {};
		\node [style=white node] (26) at (24, 1) {};
		\node [style=white node] (27) at (25, 1) {};
		\node [style=white node] (28) at (20, 1) {};
		\node [style=white node] (29) at (18, 1) {};
		\node [style=none] (30) at (-6, 0) {};
		\node [style=none] (31) at (-5, 0) {};
		\node [style=none] (32) at (-4, 0) {};
		\node [style=none] (33) at (-3, 0) {};
		\node [style=none] (34) at (0, 0) {};
		\node [style=none] (35) at (1, 0) {};
		\node [style=none] (36) at (2, 0) {};
		\node [style=none] (37) at (3, 0) {};
		\node [style=none] (38) at (6, 0) {};
		\node [style=none] (39) at (7, 0) {};
		\node [style=none] (40) at (8, 0) {};
		\node [style=none] (41) at (9, 0) {};
		\node [style=none] (42) at (12, 0) {};
		\node [style=none] (43) at (13, 0) {};
		\node [style=none] (44) at (14, 0) {};
		\node [style=none] (45) at (15, 0) {};
		\node [style=none] (46) at (18, 0) {};
		\node [style=none] (47) at (19, 0) {};
		\node [style=none] (48) at (20, 0) {};
		\node [style=none] (49) at (21, 0) {};
		\node [style=none] (50) at (24, 0) {};
		\node [style=none] (51) at (25, 0) {};
		\node [style=none] (52) at (26, 0) {};
		\node [style=none] (53) at (27, 0) {};
		\node [style=none] (84) at (-1.5, 1) {$+$};
		\node [style=none] (85) at (4.5, 1) {$+$};
		\node [style=none] (86) at (22.5, 1) {$+$};
		\node [style=none] (87) at (16.5, 1) {$+$};
		\node [style=none] (88) at (10.5, 1) {$+$};
		\node [style=black node] (89) at (-4, 4) {};
		\node [style=black node] (92) at (-3, 4) {};
		\node [style=new style 0] (93) at (-4, 2) {};
		\node [style=new style 0] (94) at (-5, 2) {};
		\node [style=black node] (97) at (1, 4) {};
		\node [style=black node] (98) at (3, 4) {};
		\node [style=new style 0] (99) at (2, 2) {};
		\node [style=new style 0] (100) at (1, 2) {};
		\node [style=black node] (101) at (8, 4) {};
		\node [style=black node] (103) at (7, 4) {};
		\node [style=new style 0] (105) at (8, 2) {};
		\node [style=new style 0] (106) at (7, 2) {};
		\node [style=black node] (108) at (12, 4) {};
		\node [style=black node] (110) at (15, 4) {};
		\node [style=new style 0] (111) at (14, 2) {};
		\node [style=new style 0] (112) at (13, 2) {};
		\node [style=black node] (113) at (20, 4) {};
		\node [style=black node] (114) at (18, 4) {};
		\node [style=new style 0] (117) at (20, 2) {};
		\node [style=new style 0] (118) at (19, 2) {};
		\node [style=black node] (120) at (24, 4) {};
		\node [style=black node] (121) at (25, 4) {};
		\node [style=new style 0] (123) at (26, 2) {};
		\node [style=new style 0] (124) at (25, 2) {};
	\end{pgfonlayer}
	\begin{pgfonlayer}{edgelayer}
		\draw (6) to (30.center);
		\draw (7) to (31.center);
		\draw (21) to (32.center);
		\draw (20) to (33.center);
		\draw (8) to (34.center);
		\draw (22) to (35.center);
		\draw (9) to (36.center);
		\draw (23) to (37.center);
		\draw (10) to (38.center);
		\draw (24) to (39.center);
		\draw (25) to (40.center);
		\draw (11) to (41.center);
		\draw (18) to (42.center);
		\draw (12) to (43.center);
		\draw (13) to (44.center);
		\draw (19) to (45.center);
		\draw (29) to (46.center);
		\draw (14) to (47.center);
		\draw (28) to (48.center);
		\draw (15) to (49.center);
		\draw (26) to (50.center);
		\draw (27) to (51.center);
		\draw (16) to (52.center);
		\draw (17) to (53.center);
		\draw (89) to (94);
		\draw [style=red-edge] (92) to (94);
		\draw (92) to (93);
		\draw (89) to (93);
		\draw (97) to (99);
		\draw [style=red-edge] (98) to (100);
		\draw [style=red-edge] (97) to (100);
		\draw (98) to (99);
		\draw (101) to (106);
		\draw (103) to (105);
		\draw [style=red-edge] (103) to (106);
		\draw (101) to (105);
		\draw (108) to (112);
		\draw [style=red-edge] (110) to (112);
		\draw (108) to (111);
		\draw (110) to (111);
		\draw (113) to (118);
		\draw (114) to (118);
		\draw (114) to (117);
		\draw (113) to (117);
		\draw (120) to (124);
		\draw (121) to (123);
		\draw [style=red-edge] (121) to (124);
		\draw (120) to (123);
	\end{pgfonlayer}
\end{tikzpicture}}$$
Note that there may be different equivalent representations of Bell states.

\paragraph{Polarising beam splitters.}
On bulk optics, the polarization states of photons can be acted upon using
wave plates, polarizing beam splitters (PBSs) and photon counting measurements.
Wave plates are simply $X$ phase rotations, represented as red nodes in $\bf{ZX}$.
The PBS admits no description in $\bf{ZX}$. It does however have a simple
interpretation in $\bf{LO}$:
$$\scalebox{0.8}
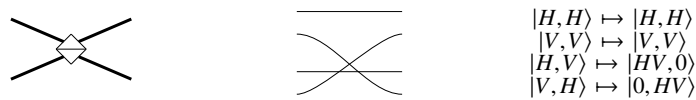
{\begin{tikzpicture}
	\begin{pgfonlayer}{nodelayer}
		\node [style=none] (8) at (-9.75, 0.25) {};
		\node [style=none] (9) at (-8.75, 0.25) {};
		\node [style=none] (10) at (-9.25, 0.75) {};
		\node [style=none] (11) at (-9.25, -0.25) {};
		\node [style=none] (12) at (-9.5, 0.5) {};
		\node [style=none] (13) at (-9, 0) {};
		\node [style=none] (14) at (-9, 0.5) {};
		\node [style=none] (15) at (-9.5, 0) {};
		\node [style=none] (16) at (-11.25, 1.25) {};
		\node [style=none] (17) at (-11.25, -0.75) {};
		\node [style=none] (18) at (-7.25, 1.25) {};
		\node [style=none] (19) at (-7.25, -0.75) {};
		\node [style=none] (53) at (-1.75, 1.5) {};
		\node [style=none] (54) at (1.75, 1.5) {};
		\node [style=none] (55) at (-1.75, 0.75) {};
		\node [style=none] (56) at (1.75, -1.25) {};
		\node [style=none] (57) at (-1.75, -0.5) {};
		\node [style=none] (58) at (1.75, -0.5) {};
		\node [style=none] (59) at (-1.75, -1.25) {};
		\node [style=none] (60) at (1.75, 0.75) {};
		\node [style=none] (75) at (8.75, 1.25) {$\ket{H,H} \, \mapsto \, \ket{H,H}$};
		\node [style=none] (76) at (8.75, -0.25) {$\ket{H, V} \, \mapsto \, \ket{HV, 0}$};
		\node [style=none] (77) at (8.75, 0.5) {$\ket{V, V} \, \mapsto \, \ket{V, V}$};
		\node [style=none] (78) at (8.75, -1) {$\ket{V, H} \, \mapsto \, \ket{0, HV}$};
	\end{pgfonlayer}
	\begin{pgfonlayer}{edgelayer}
		\draw [style=thick] (16.center) to (12.center);
		\draw [style=thick] (14.center) to (18.center);
		\draw [style=thick] (13.center) to (19.center);
		\draw [style=thick] (15.center) to (17.center);
		\draw (10.center) to (8.center);
		\draw (8.center) to (11.center);
		\draw (11.center) to (9.center);
		\draw (9.center) to (10.center);
		\draw (8.center) to (9.center);
		\draw (53.center) to (54.center);
		\draw [in=-180, out=0, looseness=0.75] (55.center) to (56.center);
		\draw (57.center) to (58.center);
		\draw [in=-180, out=0, looseness=0.50] (59.center) to (60.center);
	\end{pgfonlayer}
\end{tikzpicture}}$$
In combination with $X$ states and effects, polarising beam splitters can be used
to perform post-selected fusion measurements and their transpose:
$$\scalebox{0.8}
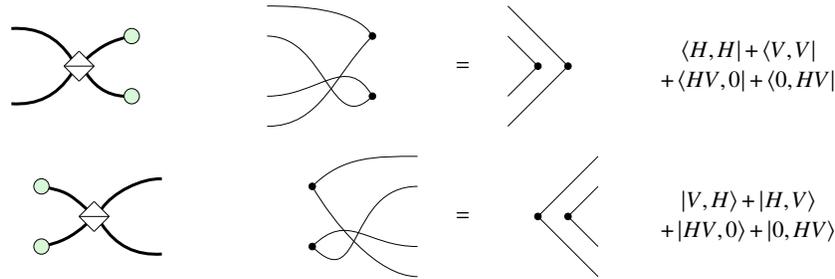
{\begin{tikzpicture}
	\begin{pgfonlayer}{nodelayer}
		\node [style=none] (0) at (-13.25, 0) {};
		\node [style=none] (1) at (-12.25, 0) {};
		\node [style=none] (2) at (-12.75, 0.5) {};
		\node [style=none] (3) at (-12.75, -0.5) {};
		\node [style=none] (4) at (-13, 0.25) {};
		\node [style=none] (5) at (-12.5, -0.25) {};
		\node [style=none] (6) at (-12.5, 0.25) {};
		\node [style=none] (7) at (-13, -0.25) {};
		\node [style=none] (8) at (-15, 1.25) {};
		\node [style=none] (9) at (-15, -1.25) {};
		\node [style=none] (10) at (-11, 1) {};
		\node [style=none] (11) at (-11, -1) {};
		\node [style=green node] (30) at (-11, 1) {};
		\node [style=green node] (31) at (-11, -1) {};
		\node [style=black node] (33) at (2.5, 0) {};
		\node [style=black node] (34) at (3.5, 0) {};
		\node [style=none] (35) at (1.5, 2) {};
		\node [style=none] (36) at (1.5, 1) {};
		\node [style=none] (37) at (1.5, -1) {};
		\node [style=none] (38) at (1.5, -2) {};
		\node [style=none] (39) at (-12, -5.25) {};
		\node [style=none] (40) at (-12.75, -5) {};
		\node [style=none] (41) at (-11.75, -5) {};
		\node [style=none] (42) at (-12.25, -4.5) {};
		\node [style=none] (43) at (-12.25, -5.5) {};
		\node [style=none] (44) at (-12.5, -4.75) {};
		\node [style=none] (45) at (-12, -5.25) {};
		\node [style=none] (46) at (-12, -4.75) {};
		\node [style=none] (47) at (-12.5, -5.25) {};
		\node [style=none] (53) at (-14, -6) {};
		\node [style=none] (58) at (-14, -4) {};
		\node [style=none] (62) at (-10, -3.75) {};
		\node [style=none] (63) at (-10, -6.25) {};
		\node [style=black node] (65) at (3.5, -5) {};
		\node [style=black node] (66) at (2.5, -5) {};
		\node [style=none] (67) at (4.5, -3) {};
		\node [style=none] (68) at (4.5, -4) {};
		\node [style=none] (69) at (4.5, -6) {};
		\node [style=none] (70) at (4.5, -7) {};
		\node [style=none] (71) at (9.5, 0.5) {$\bra{H, H} + \bra{V, V}$};
		\node [style=none] (72) at (9.5, -5.5) {$+ \ket{HV, 0} + \ket{0, HV}$};
		\node [style=none] (77) at (-3, 1) {};
		\node [style=none] (78) at (-3, 1) {};
		\node [style=none] (83) at (-3, -1) {};
		\node [style=none] (84) at (-3, -1) {};
		\node [style=black node] (97) at (-3, 1) {};
		\node [style=black node] (98) at (-3, -1) {};
		\node [style=none] (99) at (-6.5, 2) {};
		\node [style=none] (100) at (-6.5, 1) {};
		\node [style=none] (101) at (-6.5, -1) {};
		\node [style=none] (102) at (-6.5, -2) {};
		\node [style=none] (103) at (0, 0) {$=$};
		\node [style=none] (108) at (-5, -4) {};
		\node [style=none] (109) at (-5, -4) {};
		\node [style=none] (115) at (-5, -6) {};
		\node [style=none] (116) at (-5, -6) {};
		\node [style=black node] (118) at (-5, -4) {};
		\node [style=black node] (119) at (-5, -6) {};
		\node [style=none] (122) at (-1.5, -3) {};
		\node [style=none] (123) at (-1.5, -7) {};
		\node [style=none] (124) at (-1.5, -6) {};
		\node [style=none] (125) at (-1.5, -4) {};
		\node [style=none] (126) at (0, -5) {$=$};
		\node [style=none] (128) at (9.5, -4.5) {$\ket{V, H} + \ket{H, V}$};
		\node [style=none] (129) at (9.5, -0.5) {$+ \bra{HV, 0} + \bra{0, HV}$};
		\node [style=green node] (130) at (-14, -4) {};
		\node [style=green node] (131) at (-14, -6) {};
	\end{pgfonlayer}
	\begin{pgfonlayer}{edgelayer}
		\draw [style=thick, bend left] (8.center) to (4.center);
		\draw [style=thick, bend left=15] (6.center) to (10.center);
		\draw [style=thick, bend right] (5.center) to (11.center);
		\draw [style=thick, bend left] (7.center) to (9.center);
		\draw (2.center) to (0.center);
		\draw (0.center) to (3.center);
		\draw (3.center) to (1.center);
		\draw (1.center) to (2.center);
		\draw (0.center) to (1.center);
		\draw (36.center) to (33);
		\draw (33) to (37.center);
		\draw (38.center) to (34);
		\draw (34) to (35.center);
		\draw (42.center) to (40.center);
		\draw (40.center) to (43.center);
		\draw (43.center) to (41.center);
		\draw (41.center) to (42.center);
		\draw (40.center) to (41.center);
		\draw [style=thick, bend right=15] (53.center) to (47.center);
		\draw [style=thick, bend left=15, looseness=1.25] (58.center) to (44.center);
		\draw [style=thick, bend left, looseness=0.75] (46.center) to (62.center);
		\draw [style=thick, bend right] (45.center) to (63.center);
		\draw (68.center) to (65);
		\draw (65) to (69.center);
		\draw (70.center) to (66);
		\draw (66) to (67.center);
		\draw [in=-135, out=0, looseness=1.25] (100.center) to (84.center);
		\draw [in=225, out=0] (102.center) to (78.center);
		\draw [in=120, out=0, looseness=0.75] (99.center) to (77.center);
		\draw [in=120, out=0, looseness=1.25] (101.center) to (83.center);
		\draw [in=-180, out=45] (108.center) to (122.center);
		\draw [in=-180, out=-60, looseness=0.75] (109.center) to (123.center);
		\draw [in=180, out=45, looseness=1.25] (115.center) to (124.center);
		\draw [in=180, out=-45, looseness=1.50] (116.center) to (125.center);
	\end{pgfonlayer}
\end{tikzpicture}}$$
As an application, the linear-optical protocol for generating Bell states demonstrated
in \cite{zhang2008} may be described as a diagram using PBSs and ZX primitives:
$$\scalebox{0.7}
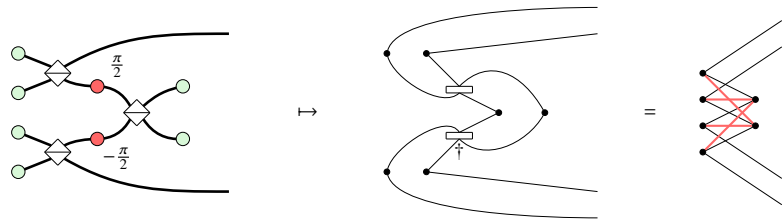
{\begin{tikzpicture}
	\begin{pgfonlayer}{nodelayer}
		\node [style=none] (0) at (-7, 4) {};
		\node [style=none] (1) at (-6, 4) {};
		\node [style=none] (2) at (-6.5, 4.5) {};
		\node [style=none] (3) at (-6.5, 3.5) {};
		\node [style=none] (4) at (-6.75, 4.25) {};
		\node [style=none] (5) at (-6.25, 3.75) {};
		\node [style=none] (6) at (-6.25, 4.25) {};
		\node [style=none] (7) at (-6.75, 3.75) {};
		\node [style=none] (8) at (-8, 5) {};
		\node [style=none] (9) at (-8, 3) {};
		\node [style=none] (10) at (-4.75, 5) {};
		\node [style=none] (11) at (-4.75, 3) {};
		\node [style=none] (22) at (-9.25, 5.25) {};
		\node [style=none] (26) at (-8, 5) {};
		\node [style=none] (27) at (-9.25, 2.75) {};
		\node [style=none] (31) at (-8, 3) {};
		\node [style=none] (38) at (-10, 2.5) {};
		\node [style=none] (39) at (-9, 2.5) {};
		\node [style=none] (40) at (-9.5, 3) {};
		\node [style=none] (41) at (-9.5, 2) {};
		\node [style=none] (42) at (-9.75, 2.75) {};
		\node [style=none] (43) at (-9.25, 2.25) {};
		\node [style=none] (44) at (-9.25, 2.75) {};
		\node [style=none] (45) at (-9.75, 2.25) {};
		\node [style=none] (46) at (-10, 5.5) {};
		\node [style=none] (47) at (-9, 5.5) {};
		\node [style=none] (48) at (-9.5, 6) {};
		\node [style=none] (49) at (-9.5, 5) {};
		\node [style=none] (50) at (-9.75, 5.75) {};
		\node [style=none] (51) at (-9.25, 5.25) {};
		\node [style=none] (52) at (-9.25, 5.75) {};
		\node [style=none] (53) at (-9.75, 5.25) {};
		\node [style=none] (54) at (-3, 7) {};
		\node [style=none] (55) at (-3, 1) {};
		\node [style=none] (60) at (-11, 3.25) {};
		\node [style=none] (63) at (-11, 3.25) {};
		\node [style=none] (65) at (-11, 1.75) {};
		\node [style=none] (68) at (-11, 1.75) {};
		\node [style=none] (70) at (-11, 4.75) {};
		\node [style=none] (73) at (-11, 4.75) {};
		\node [style=none] (75) at (-11, 6.25) {};
		\node [style=none] (78) at (-11, 6.25) {};
		\node [style=black node] (309) at (15, 4.5) {};
		\node [style=black node] (310) at (15, 2.5) {};
		\node [style=black node] (311) at (15, 3.5) {};
		\node [style=black node] (312) at (15, 5.5) {};
		\node [style=new style 0] (313) at (17, 4.5) {};
		\node [style=new style 0] (314) at (17, 3.5) {};
		\node [style=none] (315) at (18, 0.5) {};
		\node [style=none] (316) at (18, 1.5) {};
		\node [style=none] (317) at (18, 6.5) {};
		\node [style=none] (318) at (18, 7.5) {};
		\node [style=none] (325) at (11, 8) {};
		\node [style=none] (326) at (5.75, 4.75) {};
		\node [style=none] (340) at (5.25, 5) {};
		\node [style=none] (341) at (5.25, 4.75) {};
		\node [style=none] (342) at (6.25, 4.75) {};
		\node [style=none] (343) at (6.25, 5) {};
		\node [style=none] (344) at (5.75, 4.75) {};
		\node [style=none] (345) at (5.25, 3.25) {};
		\node [style=none] (346) at (5.25, 3) {};
		\node [style=none] (347) at (6.25, 3) {};
		\node [style=none] (348) at (6.25, 3.25) {};
		\node [style=none] (349) at (5.75, 3.25) {};
		\node [style=none] (350) at (5.75, 3) {};
		\node [style=none] (357) at (5.75, 5) {};
		\node [style=none] (358) at (11, 7) {};
		\node [style=none] (366) at (5.75, 3.25) {};
		\node [style=none] (367) at (11, 0) {};
		\node [style=none] (375) at (11, 1) {};
		\node [style=none] (376) at (5.75, 3) {};
		\node [style=black node] (378) at (3, 6.25) {};
		\node [style=black node] (379) at (4.5, 6.25) {};
		\node [style=black node] (391) at (9, 4) {};
		\node [style=black node] (392) at (7.25, 4) {};
		\node [style=black node] (397) at (3, 1.75) {};
		\node [style=black node] (398) at (4.5, 1.75) {};
		\node [style=none] (399) at (0, 4) {$\mapsto$};
		\node [style=none] (401) at (13, 4) {$=$};
		\node [style=green node] (402) at (-11, 1.75) {};
		\node [style=green node] (403) at (-11, 6.25) {};
		\node [style=green node] (404) at (-11, 3.25) {};
		\node [style=green node] (405) at (-11, 4.75) {};
		\node [style=red node] (406) at (-8, 5) {};
		\node [style=red node] (407) at (-8, 3) {};
		\node [style=green node] (408) at (-4.75, 5) {};
		\node [style=green node] (409) at (-4.75, 3) {};
		\node [style=none] (410) at (-7.25, 2.25) {$- \frac{\pi}{2}$};
		\node [style=none] (412) at (-7.25, 5.75) {$\frac{\pi}{2}$};
		\node [style=none] (413) at (5.75, 2.5) {$\dagger$};
	\end{pgfonlayer}
	\begin{pgfonlayer}{edgelayer}
		\draw [style=thick, bend left] (8.center) to (4.center);
		\draw [style=thick, bend left=15] (6.center) to (10.center);
		\draw [style=thick, bend right] (5.center) to (11.center);
		\draw [style=thick, bend left] (7.center) to (9.center);
		\draw (2.center) to (0.center);
		\draw (0.center) to (3.center);
		\draw (3.center) to (1.center);
		\draw (1.center) to (2.center);
		\draw (0.center) to (1.center);
		\draw [style=thick, in=180, out=-30, looseness=0.75] (22.center) to (26.center);
		\draw [style=thick, in=180, out=30, looseness=0.75] (27.center) to (31.center);
		\draw (40.center) to (38.center);
		\draw (38.center) to (41.center);
		\draw (41.center) to (39.center);
		\draw (39.center) to (40.center);
		\draw (38.center) to (39.center);
		\draw (48.center) to (46.center);
		\draw (46.center) to (49.center);
		\draw (49.center) to (47.center);
		\draw (47.center) to (48.center);
		\draw (46.center) to (47.center);
		\draw [style=thick] (68.center) to (45.center);
		\draw [style=thick] (63.center) to (42.center);
		\draw [style=thick] (73.center) to (53.center);
		\draw [style=thick] (78.center) to (50.center);
		\draw [style=thick, in=180, out=30] (52.center) to (54.center);
		\draw [style=thick, in=180, out=-30] (43.center) to (55.center);
		\draw (312) to (313);
		\draw (309) to (314);
		\draw (310) to (314);
		\draw (311) to (313);
		\draw (312) to (318.center);
		\draw (309) to (317.center);
		\draw (310) to (315.center);
		\draw (311) to (316.center);
		\draw [style=blue-edge] (312) to (314);
		\draw [style=blue-edge] (309) to (313);
		\draw [style=blue-edge] (310) to (313);
		\draw [style=blue-edge] (311) to (314);
		\draw (340.center) to (343.center);
		\draw (343.center) to (342.center);
		\draw (342.center) to (341.center);
		\draw (341.center) to (340.center);
		\draw (345.center) to (348.center);
		\draw (348.center) to (347.center);
		\draw (347.center) to (346.center);
		\draw (346.center) to (345.center);
		\draw [in=120, out=75, looseness=0.75] (397) to (366.center);
		\draw (398) to (376.center);
		\draw [in=-135, out=-75, looseness=0.75] (378) to (344.center);
		\draw (379) to (357.center);
		\draw (379) to (358.center);
		\draw [in=-180, out=75, looseness=0.50] (378) to (325.center);
		\draw (398) to (375.center);
		\draw [in=180, out=-75, looseness=0.50] (397) to (367.center);
		\draw [in=120, out=45, looseness=1.25] (357.center) to (391);
		\draw (344.center) to (392);
		\draw (366.center) to (392);
		\draw [in=-120, out=-45] (376.center) to (391);
	\end{pgfonlayer}
\end{tikzpicture}}$$
We recover the diagram for the Bell state by reducing to normal form.

\paragraph{Spiders.}
The only missing ZX generator, which we need for a complete mapping
$\bf{ZX} \to \bf{QPath}$, is the $Z$ copy spider.
We may readily deduce its representation using a known equality in ZX:
$$\scalebox{0.8}
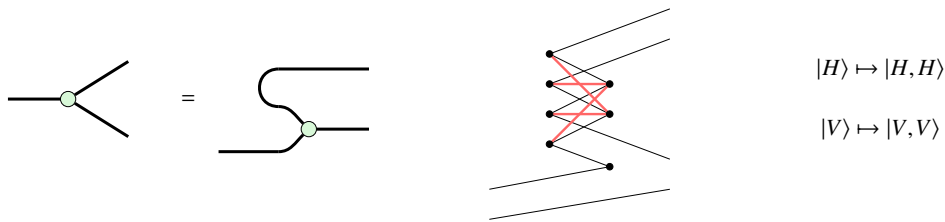
{\begin{tikzpicture}
	\begin{pgfonlayer}{nodelayer}
		\node [style=none] (21) at (-14, 1.25) {};
		\node [style=none] (22) at (-14, -1.25) {};
		\node [style=none] (23) at (-18, 0) {};
		\node [style=green node] (24) at (-16, 0) {};
		\node [style=none] (25) at (11, 1) {$\ket{H} \mapsto \ket{H,H}$};
		\node [style=none] (26) at (11, -1) {$\ket{V} \mapsto \ket{V,V}$};
		\node [style=none] (37) at (-2, -3) {};
		\node [style=black node] (38) at (0, -0.5) {};
		\node [style=black node] (39) at (0, 1.5) {};
		\node [style=black node] (40) at (0, 0.5) {};
		\node [style=black node] (41) at (0, -1.5) {};
		\node [style=new style 0] (42) at (2, -0.5) {};
		\node [style=new style 0] (43) at (2, 0.5) {};
		\node [style=none] (44) at (4, 3) {};
		\node [style=none] (45) at (4, 2) {};
		\node [style=none] (46) at (4, -2) {};
		\node [style=none] (47) at (2, -2.25) {};
		\node [style=green node] (48) at (-8, -1) {};
		\node [style=none] (50) at (-6, 1) {};
		\node [style=none] (52) at (-6, -1) {};
		\node [style=none] (53) at (-9, 1) {};
		\node [style=none] (54) at (-11, -1.75) {};
		\node [style=none] (55) at (-9, -1.75) {};
		\node [style=none] (56) at (-9, -0.25) {};
		\node [style=black node] (57) at (2, -2.25) {};
		\node [style=none] (58) at (-2, -4) {};
		\node [style=none] (59) at (4, -3) {};
		\node [style=none] (60) at (-12, 0) {};
		\node [style=none] (61) at (-12, 0) {$=$};
	\end{pgfonlayer}
	\begin{pgfonlayer}{edgelayer}
		\draw [style=thick] (23.center) to (24);
		\draw [style=thick] (24) to (21.center);
		\draw [style=thick] (24) to (22.center);
		\draw (41) to (42);
		\draw (38) to (43);
		\draw (39) to (43);
		\draw (40) to (42);
		\draw (41) to (47.center);
		\draw (38) to (46.center);
		\draw (39) to (44.center);
		\draw (40) to (45.center);
		\draw [style=blue-edge] (41) to (43);
		\draw [style=blue-edge] (38) to (42);
		\draw [style=blue-edge] (39) to (42);
		\draw [style=blue-edge] (40) to (43);
		\draw [style=thick] (52.center) to (48);
		\draw [style=thick] (53.center) to (50.center);
		\draw [style=thick, bend left=90, looseness=1.75] (56.center) to (53.center);
		\draw [style=thick] (55.center) to (54.center);
		\draw [style=thick, in=0, out=135] (48) to (56.center);
		\draw [style=thick, in=0, out=-135] (48) to (55.center);
		\draw (58.center) to (59.center);
		\draw (37.center) to (57);
	\end{pgfonlayer}
\end{tikzpicture}}$$
Similarly, we may turn the input leg into an output using a second Bell state.
This yields a protocol for generating the dual-rail GHZ state using five ancillary
photons. Note that the mapping is in no way unique, and we may obtain several equivalent
protocols by further twisiting the spider above. This however increases the number
of ancillary photons needed. As first shown in \cite{browne2005}, any cluster state
can be obtained by performing additional fusion measurements.

\section*{Outlook}

The theory in this paper is being implemented in DisCoPy \cite{defelice2020b},
the Python library for monoidal categories. DisCoPy already has a number of
tools for qubit quantum computing, including interfaces with tket \cite{sivarajah2021},
PyZX \cite{kissinger2019} and high-performance libraries for classical simulation.
DisCoPy functors will allow to compile qubit circuits and cluster states into
linear optical circuits for efficient simulation with Perceval \cite{perceval2022}
and future interfaces with photonic devices.

\section*{Acknowledgements}
The authors would like to thank Richie Yeung, Harny Wang, Douglas Brown,
Alex Cowtan and Anna Pearson also from Quantinuum in Oxford, Alexis Toumi in Paris,
Amar Hazihasanovic in Estonia, Lee Rozema and Iris Agresti in Vienna and Terry Rudolph from PsiQuantum in California,
for the many conversations on optics which led to this manuscript and the ones to come.
% The authors were also made aware of a forthcoming paper by Cl\'ement, Perdrix and
% Valiron which gives a rewriting system for an interpretation of $\bf{LO}$
% in the category of matrices with direct sum.
% The most interesting recent developments lie in the
% development of software packages such as Strawberry Fields \cite{killoran2019}
% and Perceval \cite{perceval2022}.

% The main tool for quantum photonics available on the web is Xanadu's Strawberry Fields \cite{killoran2019} which is however tailored to \emph{continuous variable} states of photons.
% Interfacing the discopy.optics module with strawberryfields would require the formalisation
% of squeezed (continuous) photon states in the bosonic Fock space. A starting
% point for this formalisation could be the definition of \emph{coherent} states given
% in \cite{vicary2008}.
% Very recently, Quandela released Perceval, their software package for \emph{discrete variable}
% quantum photonics \cite{perceval2022}. While there is a clear overlap between
% our implementation and Perceval (Sections 2 and 3 in this manuscript), the main
% difference resides in the categorical approach that we present here.
% In Sections 4 and 5 we will see the advantage of this approach, in that it
% allows for diagrammatic reasoning, rewriting and implementing the mapping
% from cluster states to photonic circuits as a discopy Functor.

\bibliographystyle{unsrt}
\bibliography{ms}

\begin{thebibliography}{10}
\providecommand{\bibitemdeclare}[2]{}
\providecommand{\surnamestart}{}
\providecommand{\surnameend}{}
\providecommand{\urlprefix}{Available at }
\providecommand{\url}[1]{\texttt{#1}}
\providecommand{\href}[2]{\texttt{#2}}
\providecommand{\urlalt}[2]{\href{#1}{#2}}
\providecommand{\doi}[1]{doi:\urlalt{http://dx.doi.org/#1}{#1}}
\providecommand{\eprint}[1]{arXiv:\urlalt{https://arxiv.org/abs/#1}{#1}}
\providecommand{\bibinfo}[2]{#2}


\bibitem{wu1950}
C.~S. Wu and I.~Shaknov.
\newblock The {{Angular Correlation}} of {{Scattered Annihilation Radiation}}.
\newblock {\em Physical Review}, 77(1):136--136, January 1950.
\newblock \doi{10.1103/PhysRev.77.136}.

\bibitem{aspect1982}
Alain Aspect, Jean Dalibard, and G{\'e}rard Roger.
\newblock Experimental {{Test}} of {{Bell}}'s {{Inequalities Using Time-Varying
  Analyzers}}.
\newblock {\em Physical Review Letters}, 49(25):1804--1807, December 1982.
\newblock \doi{10.1103/PhysRevLett.49.1804}.

\bibitem{boschi1998}
D.~Boschi, S.~Branca, F.~De~Martini, L.~Hardy, and S.~Popescu.
\newblock Experimental {{Realization}} of {{Teleporting}} an {{Unknown Pure
  Quantum State}} via {{Dual Classical}} and {{Einstein-Podolsky-Rosen
  Channels}}.
\newblock {\em Physical Review Letters}, 80(6):1121--1125, February 1998.
\newblock \doi{10.1103/PhysRevLett.80.1121}.

\bibitem{dixon2008}
A.~R. Dixon, Z.~L. Yuan, J.~F. Dynes, A.~W. Sharpe, and A.~J. Shields.
\newblock Gigahertz decoy quantum key distribution with 1 {{Mbit}}/s secure key
  rate.
\newblock {\em Optics Express}, 16(23):18790--18797, November 2008.
\newblock \doi{10.1364/OE.16.018790}.

\bibitem{zhong2020}
Han-Sen Zhong, Hui Wang, Yu-Hao Deng, Ming-Cheng Chen, Li-Chao Peng, Yi-Han
  Luo, Jian Qin, Dian Wu, Xing Ding, Yi~Hu, Peng Hu, Xiao-Yan Yang, Wei-Jun
  Zhang, Hao Li, Yuxuan Li, Xiao Jiang, Lin Gan, Guangwen Yang, Lixing You,
  Zhen Wang, Li~Li, Nai-Le Liu, Chao-Yang Lu, and Jian-Wei Pan.
\newblock Quantum computational advantage using photons.
\newblock {\em Science}, 370(6523):1460--1463, December 2020.
\newblock \doi{10.1126/science.abe8770}.

\bibitem{aaronson2010}
Scott Aaronson and Alex Arkhipov.
\newblock The {{Computational Complexity}} of {{Linear Optics}}.
\newblock {\em Proceedings of the forty-third annual ACM symposium on Theory of computing}, June 2011.
\newblock \doi{10.1145/1993636.1993682}.

\bibitem{corrielli2021}
Giacomo Corrielli, Andrea Crespi, and Roberto Osellame.
\newblock Femtosecond laser micromachining for integrated quantum photonics.
\newblock {\em Nanophotonics}, 10(15):3789--3812, 2021.
\newblock \doi{10.1515/nanoph-2021-0419}.

\bibitem{killoran2019}
Nathan Killoran, Josh Izaac, Nicol{\'a}s Quesada, Ville Bergholm, Matthew Amy,
  and Christian Weedbrook.
\newblock Strawberry {{Fields}}: {{A Software Platform}} for {{Photonic Quantum
  Computing}}.
\newblock {\em Quantum}, 3:129, March 2019.
\newblock \doi{10.22331/q-2019-03-11-129}.

\bibitem{perceval2022}
Nicolas Heurtel, Andreas Fyrillas, Grégoire {{d}e {G}liniasty}, Raphaël {{L}e
  {B}ihan}, Sébastien Malherbe, Marceau Pailhas, Boris Bourdoncle,
  Pierre-Emmanuel Emeriau, Rawad Mezher, Luka Music, Nadia Belabas, Benoît
  Valiron, Pascale Senellart, Shane Mansfield, and Jean Senellart.
\newblock Perceval: {A} {Software} {Platform} for {Discrete} {Variable}
  {Photonic} {Quantum} {Computing}.
\newblock {\em Quantum}, 7:931, February 2023.
\newblock \doi{10.22331/q-2023-02-21-931}.

\bibitem{knill2001a}
E.~Knill, R.~Laflamme, and G.~J. Milburn.
\newblock A scheme for efficient quantum computation with linear optics.
\newblock {\em Nature}, 409(6816):46--52, January 2001.
\newblock \doi{10.1038/35051009}.

\bibitem{nielsen2004}
Michael~A. Nielsen.
\newblock Optical quantum computation using cluster states.
\newblock {\em Physical Review Letters}, 93(4):040503, July 2004.
\newblock \doi{10.1103/PhysRevLett.93.040503}.

\bibitem{kok2007}
Pieter Kok, W.~J. Munro, Kae Nemoto, T.~C. Ralph, Jonathan~P. Dowling, and
  G.~J. Milburn.
\newblock Review article: {{Linear}} optical quantum computing with photonic qubits.
\newblock {\em Reviews of Modern Physics}, 79(1):135--174, January 2007.
\newblock \doi{10.1103/RevModPhys.79.135}.

\bibitem{zhang2008}
Qiang Zhang, Xiao-Hui Bao, Chao-Yang Lu, Xiao-Qi Zhou, Tao Yang, Terry Rudolph,
  and Jian-Wei Pan.
\newblock Demonstration of efficient scheme for generation of "{{Event Ready}}"
  entangled photon pairs from single photon source.
\newblock {\em Physical Review A}, 77(6):062316, June 2008.
\newblock \doi{10.1103/PhysRevA.77.062316}.

\bibitem{browne2005}
Daniel~E. Browne and Terry Rudolph.
\newblock Resource-efficient linear optical quantum computation.
\newblock {\em Physical Review Letters}, 95(1):010501, June 2005.
\newblock \doi{10.1103/PhysRevLett.95.010501}.

\bibitem{bartolucci2021}
Sara Bartolucci, Patrick Birchall, Hector Bombin, Hugo Cable, Chris Dawson,
  Mercedes {Gimeno-Segovia}, Eric Johnston, Konrad Kieling, Naomi Nickerson,
  Mihir Pant, Fernando Pastawski, Terry Rudolph, and Chris Sparrow.
\newblock Fusion-based quantum computation.
\newblock {\em arXiv:2101.09310 [quant-ph]}, January 2021.
\newblock \doi{10.1038/s41467-023-36493-1}.

\bibitem{abramsky2004}
Samson Abramsky and Bob Coecke.
\newblock A categorical semantics of quantum protocols.
\newblock In {\em Proceedings of the 19th Annual IEEE Symposium on Logic in
  Computer Science (LICS)}, pages 415--425, March 2004.
\newblock \doi{10.1109/LICS.2004.1319636}.

\bibitem{SelingerCPM}
P.~Selinger.
\newblock Dagger compact closed categories and completely positive maps.
\newblock {\em Electronic Notes in Theoretical Computer Science}, 170:139--163,
  2007.
\newblock \doi{10.1016/j.entcs.2006.12.018}.

\bibitem{CPaqPav}
B.~Coecke, {\'E}.~O. Paquette, and D.~Pavlovi{\'c}.
\newblock {Classical and quantum structuralism}.
\newblock In S.~Gay and I.~Mackie, editors, {\em Semantic Techniques in Quantum
  Computation}, pages 29--69. Cambridge University Press, 2010.
\newblock \doi{10.48550/arXiv.0904.1997}.

\bibitem{CDKZ}
B.~Coecke, R.~Duncan, A.~Kissinger, and Q.~Wang.
\newblock Strong complementarity and non-locality in categorical quantum
  mechanics.
\newblock In {\em Proceedings of the 27th Annual IEEE Symposium on Logic in
  Computer Science (LICS)}, 2012.
\newblock arXiv:1203.4988.
\newblock \doi{10.1109/LICS.2012.35}.

\bibitem{CKbook}
B.~Coecke and A.~Kissinger.
\newblock {\em Picturing Quantum Processes. A First Course in Quantum Theory
  and Diagrammatic Reasoning}.
\newblock Cambridge University Press, 2017.
\newblock \doi{10.1017/9781316219317}.

\bibitem{sivarajah2021}
Seyon Sivarajah, Silas Dilkes, Alexander Cowtan, Will Simmons, Alec Edgington,
  and Ross Duncan.
\newblock T\$|\$ket\$\textbackslash rangle\$ : {{A Retargetable Compiler}} for
  {{NISQ Devices}}.
\newblock {\em Quantum Science and Technology}, 6(1):014003, January 2021.
\newblock \doi{10.1088/2058-9565/ab8e92}.

\bibitem{kissinger2019}
Aleks Kissinger and John {van de Wetering}.
\newblock {{PyZX}}: {{Large Scale Automated Diagrammatic Reasoning}}.
\newblock {\em Electronic Proceedings in Theoretical Computer Science},
  318:229--241, May 2020.
\newblock \doi{10.4204/EPTCS.318.14}.

\bibitem{kartsaklis2021}
Dimitri Kartsaklis, Ian Fan, Richie Yeung, Anna Pearson, Robin Lorenz, Alexis
  Toumi, Giovanni {de Felice}, Konstantinos Meichanetzidis, Stephen Clark, and
  Bob Coecke.
\newblock Lambeq: {{An Efficient High-Level Python Library}} for {{Quantum
  NLP}}.
\newblock \doi{10.48550/arXiv.2110.04236}.

\bibitem{defelice2020b}
Giovanni {de Felice}, Alexis Toumi, and Bob Coecke.
\newblock {{DisCoPy}}: {{Monoidal Categories}} in {{Python}}.
\newblock {\em Electronic Proceedings in Theoretical Computer Science},
  333:183--197, Jan 2021.
\newblock \doi{10.4204/EPTCS.333.13}.

\bibitem{miranda2021quantum}
E.~R. Miranda, R.~Yeung, A.~Pearson, K.~Meichanetzidis, and B.~Coecke.
\newblock A quantum natural language processing approach to musical
  intelligence.
\newblock {\em Quantum Computer Music: Foundations, Methods and Advanced Concepts.}, 2021.
\newblock \doi{10.1007/978-3-031-13909-3\_13}.

\bibitem{coecke2008}
Bob Coecke and Ross Duncan.
\newblock Interacting {{Quantum Observables}}.
\newblock In Luca Aceto, Ivan Damg{\aa}rd, Leslie~Ann Goldberg, Magn{\'u}s~M.
  Halld{\'o}rsson, Anna Ing{\'o}lfsd{\'o}ttir, and Igor Walukiewicz, editors,
  {\em Automata, {{Languages}} and {{Programming}}}, Lecture {{Notes}} in
  {{Computer Science}}, pages 298--310, {Berlin, Heidelberg}, 2008. {Springer}.
\newblock \doi{10.1007/978-3-540-70583-3\_25}.

\bibitem{duncan2019}
Ross Duncan, Aleks Kissinger, Simon Perdrix, and John {van de Wetering}.
\newblock Graph-theoretic {{Simplification}} of {{Quantum Circuits}} with the
  {{ZX-calculus}}.
\newblock {\em Quantum}, 4:279, June 2020.
\newblock \doi{10.22331/q-2020-06-04-279}.

\bibitem{backens2021}
Miriam Backens, Hector {Miller-Bakewell}, Giovanni de~Felice, Leo Lobski, and
  John van~de Wetering.
\newblock There and back again: {{A}} circuit extraction tale.
\newblock {\em Quantum}, 5:421, March 2021.
\newblock \doi{10.22331/q-2021-03-25-421}.

\bibitem{debeaudrap2020a}
Niel {de Beaudrap} and Dominic Horsman.
\newblock The {{ZX}} calculus is a language for surface code lattice surgery.
\newblock {\em Quantum}, 4:218, January 2020.
\newblock \doi{10.22331/q-2020-01-09-218}.

\bibitem{vicary2008}
Jamie Vicary.
\newblock A categorical framework for the quantum harmonic oscillator.
\newblock {\em International Journal of Theoretical Physics},
  47(12):3408--3447, December 2008.
\newblock \doi{10.1007/s10773-008-9772-4}.

\bibitem{fiore2015axiomatics}
M.~Fiore.
\newblock An axiomatics and a combinatorial model of creation/annihilation
  operators.
\newblock {\em arXiv preprint arXiv:1506.06402}, 2015.
\newblock \doi{10.48550/arXiv.1506.06402}.

\bibitem{valiant2001a}
Leslie~G. Valiant.
\newblock Quantum computers that can be simulated classically in polynomial
  time.
\newblock In {\em Proceedings of the Thirty-Third Annual {{ACM}} Symposium on
  {{Theory}} of Computing}, {{STOC}} '01, pages 114--123, {New York, NY, USA},
  July 2001. {Association for Computing Machinery}.
\newblock \doi{10.1145/380752.380785}.

\bibitem{terhal2002b}
Barbara~M. Terhal and David~P. DiVincenzo.
\newblock Classical simulation of noninteracting-fermion quantum circuits.
\newblock {\em Physical Review A}, 65(3):032325, March 2002.
\newblock \doi{10.1103/PhysRevA.65.032325}.

\bibitem{krenn2017}
Mario Krenn, Xuemei Gu, and Anton Zeilinger.
\newblock Quantum {{Experiments}} and {{Graphs}}: {{Multiparty States}} as
  {{Coherent Superpositions}} of {{Perfect Matchings}}.
\newblock {\em Physical Review Letters}, 119(24):240403, December 2017.
\newblock \doi{10.1103/PhysRevLett.119.240403}.

\bibitem{ataman2018}
Stefan Ataman.
\newblock A graphical method in quantum optics.
\newblock {\em Journal of Physics Communications}, 2(3):035032, 2018-03.
\newblock Publisher: {IOP} Publishing.
\newblock \doi{10.1088/2399-6528/aab50f}.

\bibitem{blute1994}
R.~F. Blute, Prakash Panangaden, and R.~A.~G. Seely.
\newblock Fock {{Space}}: {{A Model}} of {{Linear Exponential Types}}, 1994.

\bibitem{defelice2019}
Giovanni {de Felice}, Amar Hadzihasanovic, and Kang~Feng Ng.
\newblock A diagrammatic calculus of fermionic quantum circuits.
\newblock {\em Logical Methods in Computer Science ; Volume 15}, page Issue 3 ;
  18605974, 2019.
\newblock \doi{10.23638/LMCS-15(3:26)2019}.

\bibitem{CK}
B.~Coecke and A.~Kissinger.
\newblock {The compositional structure of multipartite quantum entanglement}.
\newblock In {\em Automata, Languages and Programming}, Lecture Notes in
  Computer Science, pages 297--308. Springer, 2010.
\newblock \doi{10.1007/978-3-642-14162-1\_25}.

\bibitem{hadzihasanovic2015diagrammatic}
A.~Hadzihasanovic.
\newblock A {{Diagrammatic Axiomatisation}} for {{Qubit Entanglement}}.
\newblock In {\em Proceedings of the 30th {{Annual ACM}}/{{IEEE Symposium}} on
  {{Logic}} in {{Computer Science}}}, {{LICS}} '15, pages 573--584. IEEE, 2015.
\newblock \doi{10.1109/LICS.2015.59}.

\bibitem{hadzihasanovic2017algebra}
A.~Hadzihasanovic.
\newblock {\em The Algebra of Entanglement and the Geometry of Composition}.
\newblock PhD thesis, University of Oxford, 2017.
\newblock \doi{10.48550/arXiv.1709.08086}.

\bibitem{clement2020}
Alexandre Cl{\'e}ment and Simon Perdrix.
\newblock {PBS-Calculus: A Graphical Language for Coherent Control of Quantum
  Computations}.
\newblock In Javier Esparza and Daniel Kr{\'a}ľ, editors, {\em 45th
  International Symposium on Mathematical Foundations of Computer Science (MFCS
  2020)}, volume 170 of {\em Leibniz International Proceedings in Informatics
  (LIPIcs)}, pages 24:1--24:14, Dagstuhl, Germany, 2020. Schloss
  Dagstuhl--Leibniz-Zentrum f{\"u}r Informatik.
\newblock \doi{10.4230/LIPIcs.MFCS.2020.24}.

\bibitem{mccloud2022}
Paul McCloud.
\newblock The {{Category}} of {{Linear Optical Quantum Computing}}.
\newblock {\em arXiv:2203.05958 [quant-ph]}, March 2022.
\newblock \doi{10.48550/arXiv.2203.05958}.

\bibitem{clement2022}
Alexandre Clément, Nicolas Heurtel, Shane Mansfield, Simon Perdrix, and
  Benoît Valiron.
\newblock {LOv}-calculus: A graphical language for linear optical quantum
  circuits.
\newblock In Stefan Szeider, Robert Ganian and Alexandra Silva, editors, {\em 47th
International Symposium on Mathematical Foundations of Computer Science (MFCS 2022)},
volume 241 of {\em Leibniz International Proceedings in Informatics
  (LIPIcs)}, pages 35:1--35:16, Dagstuhl, Germany, 2022. Schloss
  Dagstuhl--Leibniz-Zentrum f{\"u}r Informatik.
\newblock \doi{10.4230/LIPIcs.MFCS.2022.35}.

\bibitem{bombin2021}
Hector Bombin, Chris Dawson, Ryan~V. Mishmash, Naomi Nickerson, Fernando
  Pastawski, and Sam Roberts.
\newblock Logical blocks for fault-tolerant topological quantum computation.
\newblock {\em PRX Quantum 4, 020303}, April 2023.
\newblock \doi{10.1103/PRXQuantum.4.020303}.

\bibitem{henault2015}
Francois Henault.
\newblock Quantum physics and the beam splitter mystery.
\newblock {\em The Nature of Light: What are Photons? VI}, 9570:199--213, 2015.
\newblock \doi{10.1117/12.2186291}.

\bibitem{reck1994}
Michael Reck, Anton Zeilinger, Herbert~J. Bernstein, and Philip Bertani.
\newblock Experimental realization of any discrete unitary operator.
\newblock {\em Physical Review Letters}, 73(1):58--61, July 1994.
\newblock \doi{10.1103/PhysRevLett.73.58}.

\bibitem{clements2016}
William~R. Clements, Peter~C. Humphreys, Benjamin~J. Metcalf, W.~Steven
  Kolthammer, and Ian~A. Walmsley.
\newblock Optimal design for universal multiport interferometers.
\newblock {\em Optica}, 3(12):1460--1465, December 2016.
\newblock \doi{10.1364/OPTICA.3.001460}.

\bibitem{griffiths1962}
David~J Griffiths.
\newblock {\em Introduction to electrodynamics}.
\newblock Prentice Hall New Jersey, 1962.
\newblock \doi{10.1119/1.4766311
}.

\bibitem{pirashvili2001}
Teimuraz Pirashvili.
\newblock On the {{PROP}} corresponding to bialgebras, October 2001.
\newblock \doi{10.48550/arXiv.math/0110014}.

\bibitem{bonchi2014}
Filippo Bonchi, Pawe{\l} Soboci{\'n}ski, and Fabio Zanasi.
\newblock A {{Categorical Semantics}} of {{Signal Flow Graphs}}.
\newblock In Paolo Baldan and Daniele Gorla, editors, {\em {{CONCUR}} 2014
  \textendash{} {{Concurrency Theory}}}, Lecture {{Notes}} in {{Computer
  Science}}, pages 435--450, {Berlin, Heidelberg}, 2014. {Springer}.
\newblock \doi{10.1007/978-3-662-44584-6\_30}.

\bibitem{heunen2013}
Chris Heunen.
\newblock On the {{Functor}} {$\mathscr{l}$}2.
\newblock In Bob Coecke, Luke Ong, and Prakash Panangaden, editors, {\em
  Computation, {{Logic}}, {{Games}}, and {{Quantum Foundations}}. {{The Many
  Facets}} of {{Samson Abramsky}}: {{Essays Dedicated}} to {{Samson Abramsky}}
  on the {{Occasion}} of {{His}} 60th {{Birthday}}}, Lecture {{Notes}} in
  {{Computer Science}}, pages 107--121. {Springer Berlin Heidelberg}, {Berlin,
  Heidelberg}, 2013.
\newblock \doi{10.1007/978-3-642-38164-5\_8}.

\bibitem{fock1932}
V.~Fock.
\newblock {Konfigurationsraum und zweite Quantelung}.
\newblock {\em Zeitschrift f\"ur Physik}, 75(9):622--647, September 1932.
\newblock \doi{10.1007/BF01344458}.

\bibitem{vandewetering2020}
John {van de Wetering}.
\newblock {{ZX-calculus}} for the working quantum computer scientist.
\newblock {\em arXiv:2012.13966 [quant-ph]}, December 2020.
\newblock \doi{10.48550/arXiv.2012.13966}.

\end{thebibliography}

\end{document}